\documentclass[11pt]{article}
\usepackage{bbold}
\usepackage{tikz}
\usetikzlibrary{calc}
\usepackage{authblk}

\usepackage{amsthm}
\newtheorem{lemma}{Lemma}
\newtheorem{prop}{Proposition}

\usepackage{geometry}
\usepackage[english]{babel}
\usepackage[utf8]{inputenc}
\usepackage[T1]{fontenc}
\usepackage{indentfirst}
\usepackage{amsmath}
\usepackage{amssymb}
\usepackage{eufrak}
\usepackage{graphicx}
\usepackage{psfrag}

\allowhyphens

\usepackage[small]{caption}

\newcommand{\complex}{\mathbb{C}}

\newcommand{\valos}{\mathbb{R}}
\newcommand{\eps}{\varepsilon}

\newtheorem{thm}{Theorem}

\newcommand{\ket}[1]{{\left|#1\right\rangle}}
\newcommand{\bra}[1]{{\left\langle #1\right|}}


\usepackage{ifpdf}

\ifpdf
\usepackage{epstopdf}
\usepackage[pdftex,colorlinks,urlcolor=blue,citecolor=blue,linkcolor=blue]{hyperref}
\else
\usepackage[hypertex,colorlinks,urlcolor=blue,citecolor=blue,linkcolor=blue]{hyperref}
\fi
\begin{document}

   \tikzset{
      psi/.pic={
            \draw [thick] (0.5,1) to (0.7,1);
            \draw [thick] (0.7,0.85) rectangle (2.3,1.15);
              \draw [thick] (2.3,1) to (2.5,1);
        }
    }

   \tikzset{
      psiny/.pic={
            \draw [<-,thick] (0.5,1) to (0.7,1);
            \draw [thick] (0.7,0.85) rectangle (2.3,1.15);
              \draw [thick] (2.3,1) to (2.5,1);
        }
    }
    
    \tikzset{
      omega/.pic={
              \draw [thick] (-0.5,1) to (-0.15,1);
           \draw [thick] (0.5,1) to (0.15,1);
        \draw [thick] (-0.15,0.85) rectangle (0.15,1.15);
        \draw [->,thick] (0,1.15) to  (0,\ymax);
        \node at (0,0.35) {$\omega$};
        }
      }

    \tikzset{
      omegany/.pic={
              \draw [<-,thick] (-0.5,1) to (-0.15,1);
           \draw [thick] (0.5,1) to (0.15,1);
        \draw [thick] (-0.15,0.85) rectangle (0.15,1.15);
        \draw [->,thick] (0,1.15) to  (0,\ymax);
        \node at (0,0.35) {$\omega$};
        }
      }
      
\numberwithin{equation}{section}

\title{Integrable Matrix Product States from boundary integrability}
\author[1,2]{Bal\'azs Pozsgay}
\author[3]{Lorenzo Piroli}
\author[4]{Eric Vernier}

\affil[1]{Department of Theoretical Physics, Budapest University
	of Technology and Economics, 1111 Budapest, Budafoki \'{u}t 8, Hungary}
\affil[2]{BME Statistical Field Theory Research Group, Institute of Physics,
	Budapest University of Technology and Economics, 1111 Budapest, Budafoki \'{u}t 8, Hungary}
\affil[3]{Max-Planck-Institut f\"ur Quantenoptik, Hans-Kopfermann-Str. 1, 85748 Garching, Germany}
\affil[4]{The Rudolf Peierls Centre for Theoretical Physics, Oxford University, Oxford, OX1 3NP, United Kingdom.}
\date{}

\maketitle

\abstract{We consider integrable Matrix Product States (MPS) in integrable
  spin chains and show
    that they
  correspond to ``operator valued'' solutions of the so-called twisted
  Boundary Yang-Baxter (or reflection) equation.
We argue that the integrability condition is equivalent to a new
linear intertwiner relation, which we call the ``square root
relation'', because it involves half of the steps of the
reflection equation. It is then shown that the square root relation leads
to the full Boundary Yang-Baxter equations.
We provide explicit solutions in a number of cases characterized by 
special symmetries. These correspond to the ``symmetric pairs''
$(SU(N),SO(N))$ and
$(SO(N),SO(D)\otimes SO(N-D))$,
where in each pair the first 
and second elements are the symmetry groups of the spin chain and the 
integrable state, respectively.
These solutions can be considered
as explicit representations of the corresponding twisted Yangians,
that are new in a number of cases.
Examples include certain concrete MPS relevant
for the computation of one-point functions in defect AdS/CFT.}

\section{Introduction}

In the past decade considerable effort was devoted to the
study of non-equilibrium dynamics of integrable models \cite{nonequilibrium-intro-review}. One of the
main questions was equilibration and thermalization in closed
integrable systems
\cite{rigol-quench-review,essler-fagotti-quench-review}, and a common 
setting to study these problems has 
been the quantum quench. By definition, quenching is a sudden change
of certain parameters of the Hamiltonian, and in the simplest case
this means that the ground state of some other (integrable or
non-integrable) model is released to evolve according to the
post-quench integrable Hamiltonian.
However, the condition of having a
concrete pre-quench Hamiltonian can be relaxed, and we can also study
time evolution started from other well-defined states, which are
experimentally realizable. 

Regarding global quenches one of the central questions was whether the
system relaxes to a steady state that can be described by the
so-called Generalized Gibbs Ensemble (GGE) \cite{rigol-gge}. Whereas there is no
model-independent answer to this question, the answer was found to be
positive in models equivalent to free bosons/fermions \cite{ising-quench-0,ising-quench-1,ising-quench-2,ising-quench-3,essler-truncated-gge,spyros-calabrese-gge} and in the
XXZ spin chain, which is one of the simplest yet most important
interacting solvable systems \cite{JS-CGGE}. The status of the GGE for
models with higher rank symmetries has not yet been clarified. 

The main point of the GGE is that it can describe the steady states
arising after quenches from initial states satisfying the cluster
decomposition principle. However, its predicting abilities are
somewhat limited by its generality. In order to give exact predictions
one needs to know the Lagrange-multipliers associated to all conserved
charges. These are independent and need to be computed separately (for
the precise statements in the XXZ chain see \cite{JS-CGGE,enej-gge,sajat-eric}).

On the other hand, a different approach was suggested in
\cite{sajat-integrable-quenches} (see also
\cite{davide,delfino-quench,schuricht-quench}): one should study
quenches from 
those initial states that show certain signs of integrability
themselves, thus making the exact computation of the time evolution
more tractable. This class of states was called ``integrable initial
states'',
and a number of unifying properties were collected in
\cite{sajat-integrable-quenches}. Perhaps the most important one is that the overlaps between the eigenstates and
the initial states can be expressed in a simple factorized form
\cite{sajat-neel,Caux-Neel-overlap1,sajat-minden-overlaps,zarembo-neel-cite1,zarembo-neel-cite3,ADSMPS2,adscft-kristjansen-2017,ads-neel-cite-kristjansen-so5,kristjansen-proofs}. 
This allows for an exact treatment of the quantum quench, through the
Quench Action \cite{quench-action,caux-stb-LL-BEC-quench,JS-QA-intro} or the Quantum Transfer
Matrix \cite{sajat-integrable-quenches} methods. In practice the exact
solvability means that there is no need to compute the Lagrange
multipliers of the GGE  (or Bethe root densities of the various
particle types)  separately.

The factorized overlaps with the integrable states are non-zero only for parity symmetric eigenstates (to be more
precise, for Bethe states which are eigenstates of the space reflection operator). On the
level of Bethe rapidities this leads to the ``pair structure'': the
set of rapidities consists of pairs with opposite sign.
Also, it follows that these initial states are
annihilated by all odd (under space reflection) conserved charges of
the model.
It was argued in \cite{sajat-integrable-quenches}
that this annihilation property is the most general unifying feature
of integrable states, and it should be used as a definition. An important 
reason for choosing the annihilation property as a definition was
that it is relatively easy to check it for a specific state.

In  \cite{sajat-integrable-quenches} it was also shown that in spin
chains there is a
direct relation between integrable initial states and the theory of
integrable boundaries: every solution to the so-called Boundary
Yang-Baxter (BYB) relations can be used to construct a two-site product state which
is integrable according to the new definition. This relation can be
considered the lattice counterpart of a similar correspondence in
integrable QFT, which was worked out in the pioneering work of Ghoshal
and Zamolodchikov \cite{Ghoshal:1993tm}. The work 
\cite{sajat-integrable-quenches} only treated models with crossing
symmetry (mainly the XXZ chain and its higher spin counterparts),
where the elements of the so-called $K$-matrices are translated into
the real space two-site block using the crossing matrix. In the XXZ
chain all two-site states can be obtained this way. Our further papers
\cite{sajat-su3-1,sajat-su3-2} studied two-site states in the $SU(3)$-symmetric chain, which lacks crossing
symmetry. In such models there are two types of integrable boundary
constructions, that are obtained from the conventional (``untwisted'') and the twisted
Boundary Yang-Baxter relation. In \cite{sajat-su3-1,sajat-su3-2} we showed that the local two-site
states are always related to solutions of the latter. 

An important consequence of the relation with boundary integrability
is that the Loschmidt amplitude can be computed using analytic methods
\cite{sajat-BQTM,sajat-Loschm,sajat-nonanalytic,andrei-loschmidt}.
On the one hand, this leads to an independent method to find the 
rapidity distribution functions characterizing the post-quench steady 
state \cite{sajat-Loschm,sajat-su3-1,sajat-su3-2}. On the other hand, it allows one to prove the 
validity of important properties of the latter, such as the $Y$-system 
relations for the rapidity distribution functions \cite{sajat-Loschm}, which were 
previously only conjectured \cite{jacopo-massless-1}.
Furthermore, this
correspondence can be used to determine the thermodynamic part of the
overlaps, which in the XXZ chain led to a conjecture for the overlaps with
arbitrary two-site states \cite{sajat-minden-overlaps}.

Integrable initial states have been studied also in QFT's, see
for example
  \cite{davide,Gabor-initial-states-QFT-quench,viti-delfino,Cardy-boundary-states-flow,gabor-quench-overlaps,lencses-viti-takacs}.
   An
  early finite volume overlap formula (having the same structure as
  the later results) already appeared in
  \cite{sajat-marci-boundary}.

In a completely independent line of development, integrable initial
states were also found in the context of the AdS/CFT conjecture
\cite{zarembo-neel-cite1,zarembo-neel-cite3,ADSMPS2,adscft-kristjansen-2017,ads-neel-cite-kristjansen-so5,kristjansen-proofs}. 
One-point functions in theories with a defect are given by overlaps
with eigenstates of some local spin chain 
and certain Matrix Product States (MPS) constructed from the
generators of group symmetries. It was found that
all corresponding overlaps display the pair structure, even in the higher rank
models that are solved by the nested Bethe Ansatz. Moreover, the
overlaps have a factorized form: they are given by polynomials
of $Q$-functions and a ratio of the so-called Gaudin-like
determinants. The pair structure implies integrability of these
states, which was proven by an independent method in  \cite{kristjansen-proofs}.

Despite the increasing amount of information about the new integrable
MPS their relation with boundary integrability remained unknown.
In \cite{sajat-integrable-quenches} it was shown that in
the XXZ chain all known integrable MPS are obtained by the action of
transfer matrices on two-site states, thus they fit into the 
framework of the BYB. However, in models with higher rank symmetries
this question remained unanswered.

This is the problem that we intend to treat in the present paper. We
study the connection between the integrability 
condition and the twisted BYB relation, and 
construct new explicit solutions
that produce
the integrable states found in AdS/CFT. As a byproduct, we also obtain
new integrable MPS's.

Our solutions of the twisted BYB involve an additional degree of
freedom corresponding to the auxiliary space of 
the MPS, thus they are ``operator-valued solutions''. The twisted BYB
relation (supplied with a certain symmetry relation) serves as the
defining relation of the twisted Yangian
\cite{olshanskii2,molev-yangians-review}. Our solutions are
thus explicit representations of the twisted Yangians of various types.
The representation theory of these twisted Yangians has been studied
 in detail in
\cite{molev-representations,Vidas-BCD1,Vidas-BCD2}. In the main text
we explain how our results fit into this framework and show that our
explicit realizations are new in many cases. 

The structure of the paper is as follows. In Section \ref{sec:iMPS} we
introduce the models and some examples for the integrable MPS to be
studied. In \ref{sec:inte} we investigate the integrability condition
and show how it is related to the twisted Boundary Yang-Baxter
relation. Here we also introduce a new linear intertwining relation
called the ``square root relation''. In \ref{sec:twisted} we explain
the connection between our formalism and the defining relations of the
twisted Yangian, and we also discuss the known results about the
representations of the latter object. Sections \ref{sec:su} and
\ref{sec:so} include explicit solutions for the intertwining
relation; these solutions describe the integrable MPS introduced
earlier. We conclude in \ref{sec:conclusions}, and two simple
computations are described in Appendices \ref{sec:C1}
and \ref{sec:XXZ}.

\section{Integrable MPS}

\label{sec:iMPS}

In this work we deal with integrable lattice models associated to
specific Yang-Baxter algebras, possessing certain group symmetries.
The central object is the fundamental $R$-matrix acting on the tensor product of
two vector spaces $V_1\otimes V_2$. We will be interested in cases
when the $R$-matrix is symmetric with respect to of classical
Lie-group $\mathcal{G}$, and the local physical vector spaces carry the
defining representation.
We will consider
the $\mathcal{G}=SU(N)$ invariant case \cite{Yang-nested,zam-zam}
\begin{equation}
  \label{SUNR}
R(u)=\frac{P+u}{1+u}
\end{equation}
and the $SO(N)$ invariant case  \cite{zam-zam}
\begin{equation}
  \label{SONR}
 R(u)=\frac{(u+c)P+u(u+c)-uK}{(u+c)(u+1)}.
\end{equation}
Here $P$ is the permutation operator, and $K$ is the so-called trace
operator with matrix elements $K^{cd}_{ab}=\delta_{ab}\delta^{cd}$ and
$c=N/2-1$. We limit ourselves to these two series of symmetry groups, because their
relevance to quantum many body physics
\cite{KiserletiOsszefogl-batchelor-foerster1}\footnote{
Continuum models with $SU(N)$ symmetry have been studied in
experiments with cold
atomic gases \cite{fermi-gases-experiments-review}. In this paper we only treat lattice models, but they can
be used as a starting point towards the continuum models, which
can be obtained by certain scaling limits.
} and to the AdS/CFT
correspondence \cite{zarembo-SO6}.

Both $R$-matrices satisfy the unitarity condition
\begin{equation}
 R(u) R(-u)= \mathbb{1}
\end{equation}
and the
Yang-Baxter relation
\begin{equation}
  \label{YB0}
  R_{23}(v-z) R_{13}(u-z)   R_{12}(u-v)=
     R_{12}(u-v)  R_{13}(u-z)   R_{23}(v-z),
   \end{equation}
   which is understood as an equation in $End(V_1\otimes V_2\otimes
   V_3)$ and $R_{jk}$ acts on the spaces $j$ and $k$.
   
We will also frequently use the matrix $\check
R(u)=PR(u)$, which satisfies
 \begin{equation}
   \label{YB}
  \check R_{12}(v-z) \check R_{23}(u-z)  \check R_{12}(u-v)=
    \check R_{23}(u-v) \check R_{12}(u-z)  \check R_{23}(v-z).
\end{equation}
Our $R$-matrices also satisfy the permutation symmetry
\begin{equation}
  \label{Psym}
  PR(u)P=R(u).
\end{equation}
We put forward that our methods and results can be extended to models where
\eqref{Psym} does not hold, for example to the so-called Perk-Schulz
model associated to the quantum group $U_q(sl(N))$.
Nevertheless here we assume  \eqref{Psym}, which will
imply some minor simplifications in some of the computations.

Group invariance of the $R$-matrices means that for any $G\in\mathcal{G}$
\begin{equation}
  \label{Rgroup}
  (G_1\otimes G_2) R(u)=   R(u)   (G_1\otimes G_2).
\end{equation}
Here it is understood that the group element $G$ acts in the defining
representation in the spaces $V_1$ and $V_2$. Sometimes we will
loosely denote by $\mathcal{G}$ the defining representation as well.

A special role will be played by the partial transpose operation. In
both cases \eqref{SUNR} and \eqref{SONR} we define
\begin{equation}
  R^T(u)\equiv R^{T_1}(u)=R^{T_2}(u),
\end{equation}
where $T_{1,2}$ denote partial transposition with respect to the first
or second vector spaces; their equality follows from the explicit
forms of the $R$-matrices. The group invariance properties of $R^T(u)$
are obtained from \eqref{Rgroup} after partial transposition, for example
\begin{equation}
  \label{Rgroup2}
  (G_1\otimes (G_2^T)^{-1}) R^T(u) =   R^T(u)   (G_1\otimes (G_2^T)^{-1}).
\end{equation}
In the $SO(N)$ symmetric case we have 
$(G^T)^{-1}=G$ for every group element, thus $R^T(u)$ is also group invariant. 
Moreover, the matrix satisfies the following crossing relation:
\begin{equation}
  \label{SONcrossing}
  R^{T}(u)=
\frac{(u+c-1)u}{(u+c)(u+1)}
  R(-u-c),\qquad c=N/2-1.
  \end{equation}
This also implies that $u=-c/2$ is a crossing
symmetric point, where
\begin{equation}
  R^{T}(-c/2)=R(-c/2).
\end{equation}
On the other hand, in the $SU(N)$-symmetric case we get the group
symmetry property
\begin{equation}
  \label{Rgroup3}
  (G_1\otimes G_2^*) R^T(u) =   R^T(u)   (G_1\otimes G_2^*),
\end{equation}
where the asterix denotes complex conjugation. We can thus interpret
that $R^T(u)$ acts on the tensor product of a defining and a 
conjugate representation of $SU(N)$.

A further important property of the $R$-matrices is the initial
condition
\begin{equation}
  \label{Rinit}
 \check R(0) = \mathbb{1},
\end{equation}
which is satisfied in both cases.

We consider spin chains of length $L$ with Hilbert spaces
$\mathcal{H}=\otimes_{j=1}^L \complex^N$.
We define the monodromy matrix of a homogeneous spin chain of length
$L$ as an operator acting on $\mathcal{H}\otimes V_0$ 
\begin{equation}
  T(u)=R_{0L}(u)\dots R_{02}(u)R_{01}(u).
\end{equation}
Here $0$ refers to the so-called auxiliary space and $V_0\approx \complex^N$. The transfer matrix
is the trace over the auxiliary space:
\begin{equation}
  t(u)=\text{Tr}_0 T(u).
\end{equation}
Due to \eqref{YB} the transfer matrices form a commuting family of operators, and they
can be used to define the local conserved charges of the model as
\begin{equation}
  \label{Qdef}
  Q_j=\left.\left(\frac{d}{du}\right)^{j-1} \log t(u)\right|_{u=0},\qquad j=2,3,\dots
\end{equation}
With these conventions $Q_j$ is a sum of local operators spanning $j$
sites at most, and specifically $Q_2$ can be identified with a local
two-site Hamiltonian. With these normalizations we get in the $SU(N)$ case
\begin{equation}
H=Q_2=\sum_{j=1}^L (P_{j,j+1}-1),
\end{equation}
whereas in the $SO(N)$ case
\begin{equation}
 H=Q_2=\sum_{j=1}^L (P_{j,j+1}-K_{j,j+1}/c-1),
\end{equation}
and periodic boundary conditions are understood in both cases.

The behaviour of the charges under space
reflection is
\begin{equation}
  \label{Qprop}
  \Pi Q_j \Pi=(-1)^j Q_j,
\end{equation}
where $\Pi$ is the space reflection operator. For further use we also
introduce the space reflected transfer matrix:
\begin{equation}
  \tilde t(u)=\Pi t(u) \Pi=\text{Tr}_0\ R_{01}(u)\dots R_{0L}(u).
\end{equation}

In this work we discuss translationally invariant Matrix Product States
(MPS's). Let us take $N$   $d$-dimensional matrices $\omega_j$,
$j=1\dots N$, acting on a further, $d$ dimensional auxiliary space $V_A$ (not
to be confused with the auxiliary space used in the definition of the
transfer matrices).
The MPS is an element of the Hilbert space of the chain
defined as
\begin{equation}
  \label{psidef}
    |\Psi_\omega\rangle=\sum_{j_1,\dots,j_L=1}^N{\rm tr}_{A}
\left[\omega_{j_L}\dots  \omega_{j_2}  \omega_{j_1}   \right]
|j_L,\dots,j_2,j_1\rangle.
\end{equation}
Here $|j_L,\dots,j_2,j_1\rangle$ are the real space basis vectors
with $j_k=1,\dots,N$, $k=1\dots L$.
A graphical interpretation of the MPS is given in
Fig. \ref{fig:mps}. It
is important that the vectors \eqref{psidef} are translationally
invariant, and the same set of matrices $\{\omega_j\}$ are used
for every finite volume $L$. In the rest of the paper we will use the
term MPS for both the set $\{\omega_j\}$ and the vectors
\eqref{psidef}, and sometimes we will use the short-hand
$\omega$ to refer to the collection of matrices $\{\omega_j\}$.

    \begin{figure}
      \centering
        \begin{tikzpicture}

\foreach \x [count=\n]in {0,1,2,3}{ 
   \begin{scope}[xshift = \x cm]
         \draw [->,thick] (5,0.8) to (5,1);
\draw [thick] (4.85,0.5) rectangle (5.15,0.8);
\draw [thick] (4.5,0.65) to (4.85,0.65);
\draw [thick] (5.15,0.65) to (5.5,0.65);
\node at (5,0) {$\omega$};
\pgfmathtruncatemacro\myy{round(4-\x)};
\node at (5,1.5) {$j_{\myy}$};
\end{scope}
\draw [->,thick] (4.6,0.65) to (4.5,0.65);
} 
\end{tikzpicture}
\caption{A pictorial representation of the Matrix Product State
  \eqref{psidef} built
  from the one-site block $\omega$. Here the outgoing indices
  $j_k=1,\dots,N$,
  $k=1,\dots,L$  represent
  the physical degrees of freedom, and the horizontal lines denote the
action of the matrices $\omega_{j_k}$. The matrices are assumed to act
from the right to the left; the arrow on the leftmost horizontal link
signals this convention.
The trace in the definition \eqref{psidef}
implies periodic boundary conditions.}
\label{fig:mps}
    \end{figure}
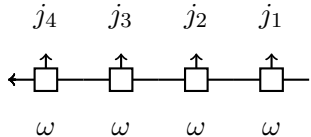

An MPS is invariant under a subgroup $\mathcal{G}'\subset \mathcal{G}$
if for every $G\in\mathcal{G}'$
\begin{equation}
  \label{groupinvariance}
\left[\otimes_{j=1}^L G_j\right]  \ket{\Psi_\omega}=\ket{\Psi_\omega}.
\end{equation}
This can be achieved by constructing a group invariant
$\omega$. Let us assume that there is a representation $\Lambda_\omega$ of $\mathcal{G}'$ acting on the
auxiliary space $V_A$. The building block $\omega$ is group invariant,
if for all $G\in\mathcal{G}'$
\begin{equation}
  \label{group1}
  \Lambda_\omega(G^{-1}) \omega_j\Lambda_\omega(G)=\sum_k G_{jk}  \omega_k,\qquad j=1,\dots,N,
\end{equation}
where $G_{jk}$ are the matrix elements of $G$ in the defining representation.
This relation immediately implies \eqref{groupinvariance}  \footnote{It can be
  shown, that given some additional assumptions the
converse is also true: If the MPS is such that
the group invariance \eqref{groupinvariance} holds in any volume $L$,
and the MPS is completely reducible (a property introduced below),
then Theorem \ref{simmm} implies the existence of a similarity
transformation which will form a representation of $G'$ on $V_A$, thus
proving the local group invariance property \eqref{group1}.}.
In the following we
will use the notation $(\mathcal{G},\mathcal{G'})$  to denote
 classes of the MPS, where it is understood that $\mathcal{G}$ and
 $\mathcal{G}'$ describe
 the group symmetries of the spin chain and the MPS,
 respectively.

The condition \eqref{group1} puts a restriction on the representation
$\Lambda_\omega$: it means that the scalar
 representation of $\mathcal{G}'$ has to be present in the
 decomposition of the
triple tensor product
 $\Lambda_\omega\otimes\bar\Lambda_\omega\otimes \Lambda_{\mathcal{G}'}$, where
 $\bar\Lambda$ is the conjugate representation to $\Lambda$ and
 $\Lambda_{\mathcal{G}'}$ is the representation of $\mathcal{G}'$
 obtained from the defining representation of $\mathcal{G}$ after the
 restriction $\mathcal{G}'\subset \mathcal{G}$.
 Examples for such group invariant
 MPS are listed below.
 
In this paper we will be interested in cases where $(\mathcal{G},\mathcal{G'})$ is a
so-called symmetric pair\footnote{A pair of Lie groups
  $(\mathcal{G},\mathcal{G'})$  is called a symmetric pair if there
  exists an involutive automorphism $\theta$ of $\mathcal{G}$ such
  that $\mathcal{G'}$ is the subgroup consisting of the
  $\theta$-invariant elements. An alternative definition on the level
  of Lie-algebras is that $(\mathfrak{g},\mathfrak{h})$ is a
  symmetric pair if $\mathfrak{g}$ can be split as a vector space into
$\mathfrak{h}\oplus\mathfrak{f}$ such that
$[\mathfrak{h},\mathfrak{h}]\in \mathfrak{h}$,
$[\mathfrak{h},\mathfrak{f}]\in \mathfrak{f}$ and
$[\mathfrak{f},\mathfrak{f}]\in \mathfrak{h}$. In this splitting
$\mathfrak{h}$ and $\mathfrak{f}$ are the eigenspaces of the
automorphism $\theta$ with eigenvalues 1 and -1, respectively.}.
The reason for concentrating on these pairs
is twofold. First, it is known that all solutions to the (twisted or
untwisted) Boundary Yang-Baxter relations are characterized by
symmetric pairs, or slight generalizations thereof (for the quantum
deformed case, see \cite{Mudrov-q,Vidas-generalized-Satake}). Second,
the known integrable initial states are characterized by symmetric
pairs. In this paper we do not attempt a rigorous classification of
all possible integrable MPS, instead we consider a few examples for
the symmetric pair $(\mathcal{G},\mathcal{G'})$.

In this work we focus on the special class of integrable MPS.
The properties of integrable states were
reviewed in detail in \cite{sajat-integrable-quenches}. It is a unifying
characteristic that they have non-zero overlaps only with the
parity-invariant Bethe states of the system,
which leads to the so-called pair structure for the Bethe
rapidities. This leads to the condition of annihilation by the odd
charges of the model:
\begin{equation}
  \label{int1}
  Q_{2j+1}\ket{\Psi}=0,\qquad j=1,\dots.
\end{equation}
It follows from \eqref{Qprop} and \eqref{Qdef} that 
(supplied with two-site shift invariance) this is
completely equivalent to
\begin{equation}
  \label{int2}
  t(u)\ket{\Psi}=\tilde t(u) \ket{\Psi}.
\end{equation}
This condition is stronger than simply parity
invariance: the relation $\Pi\ket{\Psi}=\pm \ket{\Psi}$ does not
imply \eqref{int2}. We stress that \eqref{int1} and \eqref{int2} are
exact equalities that hold in every finite volume.

In \cite{sajat-integrable-quenches} it was suggested that the
relations \eqref{int1}-\eqref{int2} should serve as definitions of
integrable states. Also, it was shown there that in the
spin-1/2 XXZ chains the known integrable MPS can be obtained from
solutions of the Boundary Yang-Baxter (BYB) relation. The connection
between integrability of the state and integrable boundaries is the following.

When having integrable boundaries one usually means boundaries in
space, and the goal is to set up a set of commuting transfer matrices, where
the so-called boundary $K$-matrices play an essential role. They
encode the boundary conditions, and they need to satisfy the so-called
reflection equations \cite{sklyanin-boundary}, guaranteeing the
commutativity of the transfer matrices. 
In contrast, in our case the integrable initial states can be considered as boundaries in
time. They can be interpreted as integrable boundaries if we build
classical 2D partition functions (corresponding to a lattice path
integral of some 1D quantum problem) and afterwards exchange the role
of space and time. 
This correspondence 
is a lattice counterpart of the picture in integrable QFT
\cite{Ghoshal:1993tm}. 
It is the purpose of the present work (and in particular Section
\ref{sec:inte}) to apply these ideas of \cite{sajat-integrable-quenches}
to the MPS in the higher rank $SU(N)$ and $SO(N)$ invariant models. 

Examples of integrable MPS were found in the context of the AdS/CFT
conjecture
\cite{zarembo-neel-cite1,zarembo-neel-cite3,ADSMPS2,adscft-kristjansen-2017,kristjansen-proofs},
and factorized overlap formulas were also derived there. 
The integrability of the MPS was proven in \cite{kristjansen-proofs} by a
method similar to our work; nevertheless the relation
with the boundary integrability framework remained unexplored. This is
the goal that we set in our paper. We
specifically focus on the following list of MPS, all of which are
integrable according to the definitions \eqref{int1}-\eqref{int2}. 

\begin{enumerate}
\item For the symmetric pair $(SU(3),SO(3))$ the MPS is given by the matrices
  \begin{equation}
    \omega_{j}=S_j, \quad j=1,2,3,
  \end{equation}
  where $S_j$ are the generators of the $SU(2)$-algebra
  \begin{equation}
  [S_j,S_k]=i\eps_{jkl}S_l
  \end{equation}
  in a finite dimensional irreducible representation.
In the simplest case of spin-1/2 representation we can choose simply
  $\omega_j=\sigma_j/2$ where $\sigma_j$ are the standard Pauli
  matrices.

The $SU(2)$-generators are in the adjoint representation, therefore
they satisfy \eqref{group1} with respect to the $SO(3)$ group.

This MPS appeared for the first time in \cite{ADSMPS2}, and the
quantum quench started from this state (in the spin-1/2 case) was
investigated in \cite{nested-quench-1}.

\item For the symmetric pair $(SU(N),SO(N))$ with any $N\ge 3$ the MPS is given by the
  matrices
   \begin{equation}
    \omega_{j}=\gamma_j,\quad j=1,\dots,N,
  \end{equation}
  where $\gamma_j$ are the $N$-dimensional gamma matrices satisfying
  the (euclidean) Clifford algebra relations
  \begin{equation}
    \{\gamma_j,\gamma_k\}=2\delta_{jk}.
  \end{equation}

The commutators of the gamma matrices can be used as generators of the $SO(N)$
Lie-algebra as
\begin{equation}
  t_{jk}=\frac{1}{4i}[\gamma_j,\gamma_k].
\end{equation}
The generators constitute the adjoint representation of
$SO(N)$, and the set of Dirac matrices satisfies \eqref{group1}
with respect to the spinor representation of $SO(N)$.

For $N=3$ this state coincides with the spin-1/2 member of the
previous family. For $N>3$ it is new.

\item For the symmetric pair $(SU(N),SO(N))$ with any $N\ge 3$ the MPS
  is given by symmetrically fused Gamma matrices
  \begin{equation}
    \omega_j=\Gamma^{(n)}_j/2,
  \end{equation}
where $\Gamma^{(n)}$ are defined in the main text in
\eqref{Gamman}. For $N=3$ they are identical to the first family
listed above, for higher $N$ they are new.
  
\item Non-trivial MPS with a different type of symmetry breaking are
  found as follows. Consider the symmetric pair $(SO(N),SO(D)\oplus
  SO(N-D))$ with some $N\ge 3$ and $1\le D \le N$ and the MPS
  \begin{equation}
    \omega_j=
    \begin{cases}
      \gamma_j & \text{ for } 1\le j \le D \\
      0 & \text{ for } D< j \le N,
    \end{cases}
  \end{equation}
where now $\gamma_j$ are the $D$-dimensional gamma matrices.

Two examples of this (for $N=6,D=3$ and $N=6,D=5$) were studied in
\cite{adscft-kristjansen-2017,ads-neel-cite-kristjansen-so5,kristjansen-proofs},
but the generalization to arbitrary $N,D$ is new.
\end{enumerate}

This list does not exhaust all the integrable MPS that were found in
the AdS/CFT literature. Symmetrically fused solutions
associated with the pairs $(SO(6),SO(3)\otimes SO(3))$ and
$(SO(6),SO(5))$ were also studied there, but in the present paper we limit ourselves to the
cases given above. 

\section{Intertwining relations}

\label{sec:inte}

In this section we analyze the integrability
condition \eqref{int2} and its relation to certain intertwining
relations and the twisted Boundary Yang-Baxter (BYB) relation. First
we need to introduce some mathematical properties of the MPS.

We call an MPS given by $\omega$ irreducible, if there is no proper subspace
$V_A'\subset V_A$ which is an invariant 
sub-space for all $\omega_{a}$, $a=1,\dots,N$. All examples listed at the end of the
previous Section are irreducible.

We will make
use of the following simple statement, which is an analogue of 
Schur's lemma: 
\begin{lemma}
If an MPS is irreducible, then any matrix $U\in \text{End}(V_A)$ which commutes with all
$\omega_{a}$ is proportional to the identity matrix
acting on $V_A$.
\label{irredU}
\end{lemma}
\begin{proof}
Any matrix $U$ has at least one eigenvector with some eigenvalue
$\lambda$. The eigenspace associated to $\lambda$ is an invariant
subspace for all $\omega_{a}$ due to the commutativity, therefore it
has to be identical to the full $V_A$.
\end{proof}

Taking two irreducible MPS $\omega_A$ and $\omega_B$ we can construct a
new one simply by addition of the auxiliary vector spaces:
\begin{equation}
  \label{dsum}
  \omega_{A+B,a}(u)\equiv
  \begin{pmatrix}
    \omega_{A,a}(u) & 0\\
0    &  \omega_{B,a}(u) 
  \end{pmatrix},\qquad
  a=1\dots N.
\end{equation}
For the physical vectors \eqref{psidef} we get
\begin{equation}
  \ket{\Psi_{A+B}}=\ket{\Psi_A}+\ket{\Psi_B}
\end{equation}
due to the additivity of the trace in.

We say that an MPS is completely reducible, if it can be written as a
direct sum of irreducible pieces.
Alternatively, it means that if there is a common invariant
subspace $V_x\in V_A$ for all $\omega_j$, then there is a
complementary space $V_y$ with $V_A=V_x\oplus V_y$, which is invariant too. 
This property
means that if we have a triangular block diagonal form as
\begin{equation}
  \omega_j=
  \begin{pmatrix}
    E_j & F_j \\
    0 & G_j
  \end{pmatrix},
\end{equation}
then necessarily $F_j\equiv 0$. Note that if we are interested in the
physical vectors, then the $F_j$ are 
irrelevant, because they do not contribute to the trace. However,
below we will be interested in linear relations involving the full
matrices without dropping any off-diagonal blocks, therefore the
complete reducibility will be crucial\footnote{There can be solutions
  to the Boundary Yang-Baxter relations that are not completely
  reducible and involve a nonzero $F_j$ block. The  treatment of these
  cases would need more detailed investigations.}.

Let $\mathcal{A}$ be the matrix algebra generated by the set of
matrices $\{\omega_j\}$. The notions of irreducibility and complete
reducibility naturally carry over to $\mathcal{A}$.
There is a well known statement which usually goes under the name of 
  Burnside's Theorem on matrix algebras:
 the algebra $\mathcal{A}$ is irreducible, if and only if
 $\mathcal{A}=End(V_A)$.
 
For any MPS we can define the associated transfer matrix\footnote{This
transfer matrix can be defined for any MPS, regardless of integrability. It is not to be confused
with the various transfer matrices built from integrable Lax
operators and integrable boundary conditions. Nevertheless the
so-called Quantum Transfer Matrix defined later in \eqref{TABdef}
simplifies to $\mathcal{T}^2$ under a special choice of parameters.} acting on $V_A\otimes
V_A$ as
\begin{equation}
  \mathcal{T}=\sum_{a=1}^N \omega_a\otimes \omega_a^*.
\end{equation}
It is known from the theory of matrix product states that the
eigenvalues and eigenvectors of $\mathcal{T}$ describe the periodicity
properties, the half-chain reduced density matrix eigenvalues,  and
the correlation lengths of the
MPS \cite{mps-intro1}.

It is shown in \cite{mps-intro1}, that if the MPS is irreducible, then there is a
non-degenerate leading eigenvalue
$\lambda_{\text{max}}$ of $\mathcal{T}$ such that $\lambda_{\text{max}}\in \valos^+$,
and  $|\lambda_j|\le \lambda_{\text{max}}$ for all other $j$. 
We say that an irreducible MPS is pure, if $\lambda_{\text{max}}$ is non-degenerate also in
magnitude:   $|\lambda_j|< \lambda_{\text{max}}$ for all other $j$ (in \cite{mps-intro1}
this was called the C2 property).
If the MPS is not pure, then there are some eigenvalues of the form
$\lambda_j=\lambda_{\text{max}}e^{2i\pi p/q}$, and the
vector \eqref{psidef} can be
written as a sum of $q$-site invariant MPS with lower bond dimension.
The purity condition thus says that the MPS can not be
decomposed into simpler blocks even if we lift the requirement of
one-site invariance. 

In the theory of Matrix Product States it is a general important
question, how unique the actual matrix
representations are. Regarding this question we have the following
theorem \cite{mps-math-similarity-1,mps-math-similarity-2}\footnote{We are thankful to Ben
  Grossmann for his help in finding the relevant literature.}:
\begin{thm}
  If two sets of matrices $\{A_j\}$ and $\{B_j\}$ are completely reducible, and they
  produce the same MPS for all $L$:
\begin{equation}
  \label{egyenloek}
  \ket{\Psi_A(L)}=\ket{\Psi_B(L)},
\end{equation}
then there
is a simultaneous similarity transformation $S$ connecting the two sets as 
  \begin{equation}
    A_j=S^{-1}B_jS,\quad j=1,\dots,N.
  \end{equation}
  \label{simmm}
\end{thm}
The proof of the theorem is given in \cite{mps-math-similarity-2}
(Theorem 1) using the representation theory of semi-groups. The theorem
also follows from Corollary 2.7 of \cite{mps-math-similarity-1}. This
latter paper also shows
that it is
enough to require the equality of MPS at some large enough length
$L^*$, the precise value of which is not important for our purposes.
A slightly different formulation of the same statement can
be extracted also from  \cite{mps-intro1}, but the theorem in
this form is not stated there.

It is easy to see that the complete reducibility is indeed needed. Let us
consider for example the simple case with $N=1$ and the two matrices
being
\begin{equation}
  A_1=
  \begin{pmatrix}
    1 & 0 \\ 0 & 1
  \end{pmatrix}\qquad
  B_1=
  \begin{pmatrix}
    1 & \alpha \\ 0 & 1
  \end{pmatrix}
\end{equation}
with some constant $\alpha\ne 0$. The resulting trace conditions
$\text{Tr}A_1^L=\text{Tr}B_1^L$ are satisfied for all $L$, but there
is no similarity transformation connecting $A$ and $B$. The reason for
this is the invariant subspace and the triangular structure, and that the
algebra generated by $B_1$ is not completely reducible.

\bigskip

Now we consider the integrability condition \eqref{int2}. Both sides of the relation \eqref{int2} can be described by a
``dressed'' MPS, where the dressing is caused by the action of the two
transfer matrices. This is made explicit by defining
the corresponding MPS
matrices $A_j$ and $B_j$ acting on $V_0\otimes V_A$, where $V_0$ is
the auxiliary space of the transfer matrix and $V_A$ is the space for
the action of $\omega_j$. Here $A_j$ corresponds to the original
dressing and $B_j$ to the reflected dressing of the MPS.
For a physical space $V_1$ let us decompose
the $R$-matrix as
\begin{equation}
  R_{10}(u)=\sum_{ab} E_{ab} \otimes \mathcal{L}_{ab}(u),
\end{equation}
where $E_{ab}$ are the basis matrices acting on $V_1$ and the matrices
$\mathcal{L}_{ab}(u)$ act on the auxiliary space. Then we have
\begin{equation}
  \label{ABdef}
  \begin{split}
    A_j&=\sum_k \mathcal{L}_{jk}(u)\otimes \omega_k \\
    B_j&=\sum_k \mathcal{L}^T_{jk}(u)\otimes \omega_k, \\
  \end{split}
\end{equation}
where $T$ denotes simple transpose for each $j,k$. 

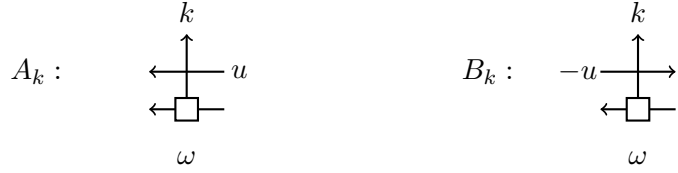
\begin{figure}
  \centering
  \begin{tikzpicture}
       \foreach \x [count=\n]in {0,6}{
    \begin{scope}[xshift=\x cm]
   \draw [<-,thick] (-0.5,1) to (-0.15,1);
           \draw [thick] (0.5,1) to (0.15,1);
        \draw [thick] (-0.15,0.85) rectangle (0.15,1.15);
        \draw [->,thick] (0,1.15) to  (0,2);
       \node at (0,0.35) {$\omega$};
        \node at (0,2.3) {$k$};
          \end{scope}}
        \draw [<-,thick] (-0.5,1.5) to (0.5,1.5);
        \draw [->,thick] (5.5,1.5) to (6.5,1.5);
        \node at (-2,1.5) {$A_k:$};
        \node at (4,1.5) {$B_k:$};
        \node at (0.7,1.5) {$u$};
        \node at (5.2,1.5) {$-u$};
      \end{tikzpicture}
  \caption{Graphical representation of the dressed matrices defined in
    \eqref{ABdef}, which act on the product space $V_0\otimes V_A$. The outgoing index $k$ stands for the physical 
    degree of freedom. The dressing of the $B_k$ matrices includes a
    partial transpose, this is denoted by the reversed arrow on the
    horizontal line. The spectral parameters associated to the
    horizontal lines are $u$ and $-u$, whereas the vertical line
    carries 0 rapidity. Here and in the following the
    $R$-matrices at the crossings are such that their argument is the
    incoming rapidity coming from the right minus the rapidity from
    the left. Therefore, both crossings above are described by $R(u)$,
    but with a different orientation.}
  \label{fig:AB}
\end{figure}

The $\omega$ matrices satisfy the group invariance \eqref{group1}
under $\mathcal{G}'\in \mathcal{G}$.
It follows that the matrices $A_j$ act on the representation
$\Lambda_0\otimes \Lambda_\omega$ of $\mathcal{G'}$, where $\Lambda_0$ is the
restriction of the defining representation of $\mathcal{G}$ to $\mathcal{G}'$.
On the other hand,
the $B_j$ act on
$\bar{\Lambda}_0\otimes \Lambda_\omega$ due to the partial transpose and the group
property \eqref{Rgroup3}.

The condition
\eqref{int2} says that $\{A_j\}$ and $\{B_j\}$ generate the same
MPS for all $L$:
\begin{equation}
  \label{egyenloek2}
  \ket{\Psi_A(L)}=\ket{\Psi_B(L)}.
\end{equation}
This implies:
\begin{thm}
  \label{Kint}
  If  the dressed MPS  $\{A_j\}$ and $\{B_j\}$ are completely
  reducible, then
  there exists an invertible matrix $K(u)$
which is a simultaneous intertwiner between $\{A_j\}$ and
$\{B_j\}$:
   \begin{equation}
    \label{intertwK}
    A_jK(u)=K(u)B_j,\quad j=1,\dots,N.
  \end{equation}
  The matrix $K(u)$ acts on $V_0\otimes V_A$, and it is a function of  the
  rapidity parameter used in the dressing \eqref{ABdef}.  
\end{thm}
\begin{proof}
This follows immediately from Theorem \ref{simmm}.
\end{proof}
The intertwiner above can be identified
with the object called ``$K$-matrix'' from the theory of boundary
integrability. Below we will show that (given some conditions) it satisfies the twisted
Boundary Yang-Baxter relation. $K(u)$ intertwines the
representations $\bar{\Lambda}_0\otimes \Lambda_\omega$ and $\Lambda_0\otimes
\Lambda_\omega$ of $\mathcal{G}'$. The representation changing property shows
that this object always corresponds to ``soliton non-preserving''
boundary conditions \cite{Delius-non-preserving,doikou-non-preserving}.

It is important to analyze the reducibility conditions for Theorem \ref{Kint} and
their implications.
First of all, in those cases when the dressed MPS are irreducible, the $K$-matrices
are unique up to an overall phase. If there are invariant subspaces,
but the MPS is completely reducible, then the
normalization factors of the individual blocks can be chosen
independently, and the $K$-matrices are thus not unique.

The above theorem requires complete reducibility, but this does not automatically hold for the dressed MPS
\eqref{ABdef}.
A counter-example in the $SU(N)$-invariant case is simply the reference state, which can be
described by a one-dimensional (scalar) MPS given by
\begin{equation}
  \omega_1=1,\qquad \omega_k=0,\text{ for } k\ge 2
\end{equation}
Let us denote by $e_k$, $k=1\dots N$ the standard basis vectors. We
can see that $span\{e_1\}$ and $span\{e_2,e_3,\dots\}$
are invariant subspaces for $\{A_j\}$ and $\{B_j\}$, respectively, but
the complements are not. In other words, both $\{A_j\}$ and $\{B_j\}$
have a non-trivial triangular structure.
However, even in this case there is a
non-zero $K$-matrix satisfying \eqref{intertwK}, which is given
explicitly as $K_{11}(u)=1$, $K_{jk}(u)=0$ for $j>1$ or $k>1$. This
$K$-matrix is not invertible, and this reflects the lack of the
complete reducibility.

For these cases we have following theorem:
\begin{thm}
  If  the original MPS is built from self-adjoint matrices
  $\omega_j=\omega_j^\dagger$, then there is at least one solution to the intertwining
  relation \eqref{intertwK} with a non-zero $K(u)$.
\end{thm}
\begin{proof}
Let us take a real $u$ parameter. The definition \eqref{ABdef}
together with the self-adjointness property implies
\begin{equation}
  A_j=B_j^\dagger, \qquad j=1,\dots,N.
\end{equation}
If there is at least one irreducible subspace $V_x$
for the set $\{A_j\}$, then $\{A_j\}$ can always be brought into an upper triangular block
form, where the first block corresponds to $V_x$. The self adjointness
implies that in this basis the set $\{B_j\}$ will have lower
triangular block form. The diagonal block corresponding to $V_x$ is already irreducible, and
the complementary block is either irreducible or can be split up into
 irreducible blocks after further basis transformations.
Only the diagonal blocks contribute to the
trace, therefore Theorem \eqref{simmm} implies a similarity transformation for each diagonal block
separately. However, this does not mean that the full sets $\{A_j\}$
and $\{B_j\}$ are similar. Nevertheless, we can find an intertwiner as
\begin{equation}
  K(u)=P_x K_x(u) P_x,
\end{equation}
where $P_x$ is the projector onto $V_x$ and $K_x(u)$ is the invertible
similarity transformation within $V_x$. If $V_x$ is a proper subspace
then $K(u)$ is not invertible, nevertheless it satisfies the linear
relation \eqref{intertwK}.

The intertwining relation is analytic in $u$, therefore the
solution $K(u)$  can
be extended into the complex plane. 
\end{proof}
This theorem obviously covers the case of the reference state, which
was a counter-example to complete reducibility.

The theorem can be extended to those cases when $span\{\omega_j\}$ is
self-adjoint, but the matrices themselves are not. This includes all
cases listed at the end of the previous Section. We thus conclude that
there is a non-zero intertwiner in all of these cases.

The uniqueness
of the intertwiner is an important question. 
We numerically investigated the dressed matrices obtained from some
examples of the MPS
listed at the end of the previous Section. 
It was found that in all cases the dressed MPS are irreducible for generic
values of $u$. The details of this numerical procedure are explained in 
Appendix \ref{sec:C1}. We have thus established that the intertwiner
is unique for these MPS, but it would be desirable to obtain an
analytic proof too.

\bigskip

Let us write the $K$-matrix introduced above as
\begin{equation}
  \label{Kpsi}
K(u)=\sum_{a,b}E_{ab}\otimes \psi_{ab}(u),
\end{equation}
where $E_{ab}$ are the elementary matrices acting on $V_0$ and $\psi_{ab}(u)$, $a,b=1\dots N$ are
matrices acting on $V_A$. The collection of matrices
$\{\psi_{ab}(u)\}_{a,b=1\dots N}$ will also be denoted simply as
$\psi(u)$ and it will be called the
two-site block. It can be considered as an element of $\complex^N\otimes
\complex^N\otimes End(V_A)$.
Later in this Section we will use $\psi(u)$ to build inhomogeneous two-site
invariant MPS. A pictorial representation is given on the left hand
side of Figure \ref{fig:init}.

    \begin{figure}
  \centering
  \begin{tikzpicture}

    \begin{scope}[shift={(-5,0)}]
  \draw [<-,thick] (0.5,0.65) to (0.7,0.65);
    \draw [thick] (2.3,0.65) to (2.5,0.65);
    \draw [->,thick] (1,0.8) to (1,1);
       \draw [->,thick] (2,0.8) to (2,1);
       \draw [thick] (0.7,0.5) rectangle (2.3,0.8);
       \node at (1.5,0) {$\psi(u)$};
            \node at (1.3,1) {$u$};
      \node at (2.3,1) {$-u$};
\end{scope}

    \draw [<-,thick] (0.5,0.65) to (0.7,0.65);
    \draw [thick] (2.3,0.65) to (2.5,0.65);
    \draw [->,thick] (1,0.8) to (1,1);
       \draw [->,thick] (2,0.8) to (2,1);
       \draw [thick] (0.7,0.5) rectangle (2.3,0.8);
\node at (1.5,0) {$\psi(0)$};
       
\node at (3,0.65) {$\sim$};

    \begin{scope}[shift={(-1,0)}]
   \draw [->,thick] (5,0.8) to (5,1);
       \draw [->,thick] (6,0.8) to (6,1);
\draw [thick] (4.85,0.5) rectangle (5.15,0.8);
\draw [thick] (5.85,0.5) rectangle (6.15,0.8);
\draw [<-,thick] (4.5,0.65) to (4.85,0.65);
\draw [thick] (5.15,0.65) to (5.85,0.65);
\draw [thick] (6.15,0.65) to (6.5,0.65);
\node at (5,0) {$\omega$};
\node at (6,0) {$\omega$};
\end{scope}
  \end{tikzpicture}
  \caption{Pictorial representation of the two-site block $\psi(u)$
    and its initial condition (see Theorem \ref{psi0}).
}
    \label{fig:init}
\end{figure}
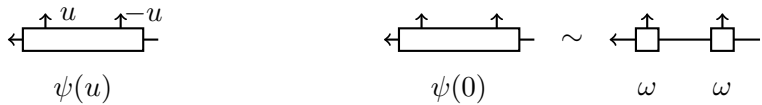

Working out the indices the condition
\eqref{intertwK} can be written as
\begin{equation}
  \label{gyokBYB00}
   \check R_{ab}^{de}(u)  \omega_d\psi_{ec}(u) =
\check R_{bc}^{de}(u)  \psi_{ad}(u) \omega_e,
\end{equation}
where $a,b,c,d,e=1\dots N$ are the physical indices and summation over
$d,e$ is understood. A detailed derivation of \eqref{gyokBYB00} is
presented in Appendix \ref{app:basszus}. This relation can also be written in a short-hand notation as
\begin{equation}
  \label{gyokBYB}
  \check R_{23}(u)  (\omega\cdot \psi(u))=
\check R_{12}(u)  (\psi(u)\cdot \omega),
\end{equation}
which is a relation in $V_3\otimes V_2\otimes V_1\otimes End(V_A)$.
Here and in the following we always understand that the $\check R$
matrices act on some of the physical indices (and the numeric indices
label the vector spaces), and the actual elements of the
equations are matrices acting on $V_A$.
A pictorial
representation is given in Fig. \ref{fig:gyokBYBmps}, whereas the
interpretation of the intertwiner relation \eqref{intertwK} is shown
in Fig. 
\ref{fig:intertw1}. We call \eqref{gyokBYB} the ``square root
relation'' (abbreviated as sq.r.r.), because it involves half the
steps of the so-called Boundary Yang-Baxter relation, to be introduced
below. We stress that \eqref{gyokBYB} is only a short notation for the
full relation \eqref{gyokBYB00} containing all indices.

\begin{thm}
 For invertible $\omega_j$ the sq.r.r. implies the initial condition 
 \begin{equation}
   \label{init}
\psi_{jk}(0) = \omega_j \omega_k
\end{equation} 
\label{psi0}
up to a scalar factor.
\end{thm}
 \begin{proof}
From \eqref{gyokBYB00} at $u=0$ we get
$$\omega_i  \psi_{j k}  = \psi_{i j} \omega_k \qquad \forall i,j,k$$ 
So we can write 
$$  
\psi_{j,k} = ((\omega_i)^{-1} \psi_{i,j} )   \omega_k = \tilde{\omega}_{j} \omega_k \,,
$$
where we have defined $\tilde{\omega}_j = (\omega_i)^{-1} \psi_{i,j} $, which has to be independent of $i$.

Substituting back we see that
 $$
 \tilde{\omega}_j {\omega_j}^{-1}  = (\omega_i)^{-1} \tilde{\omega}_{i}  \,,
 $$
 which is therefore independent both of $i,j$. Introducing $U=
 (\omega_i)^{-1} \tilde{\omega}_{i}$
 we have
 \begin{equation} 
\tilde{\omega}_j = U \omega_j  = \omega_j  U  \,, 
 \end{equation}   
 Using Lemma \eqref{irredU} $U$ is proportional to identity.
 Up to a re-scaling of $\psi(u)$ we conclude $\tilde{\omega}_j = \omega_j$, thus completing the proof.
\end{proof}

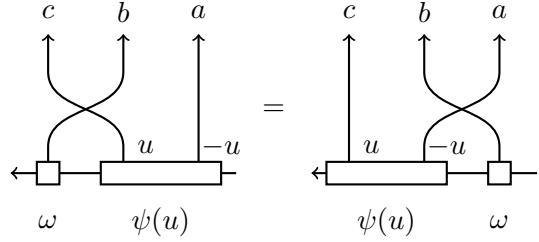
\begin{figure}
  \centering
  \begin{tikzpicture}
    \draw [<-,thick] (-0.5,0.65) to (-0.15,0.65);
    \draw [thick] (0.15,0.65) to (0.7,0.65);
    \draw [thick] (2.3,0.65) to (2.5,0.65);
      \draw [->,thick] (0,0.8) to  (0,1) to [out=90,in=-90]  (1,2) to
      (1,2.5);
      \draw [thick] (-0.15,0.5) rectangle (0.15,0.8);
\draw [thick] (0.7,0.5) rectangle (2.3,0.8);
    \draw [->,thick] (1,0.8) to (1,1) to [out=90,in=-90]  (0,2) to (0,2.5);
    \draw [->,thick] (2,0.8) to (2,2.5);
    \node at (0,0) {$\omega$};
    \node at (1.5,0) {$\psi(u)$};
    \node at (2.3,1) {$-u$};
    \node at (1.3,1) {$u$};
    \node at (0,2.8) {$c$};
    \node at (1,2.8) {$b$};
    \node at (2,2.8) {$a$};

\node at (3,1.5) {$=$};

    \draw [<-,thick] (3.5,0.65) to (3.7,0.65);
    \draw [thick] (5.3,0.65) to (5.85,0.65);
    \draw [thick] (6.15,0.65) to (6.5,0.65);
      \draw [thick] (5.85,0.5) rectangle (6.15,0.8);
\draw [thick] (3.7,0.5) rectangle (5.3,0.8);
 \draw [->,thick] (6,0.8) to  (6,1) to [out=90,in=-90]  (5,2) to    (5,2.5);
    \draw [->,thick] (5,0.8) to (5,1) to [out=90,in=-90]  (6,2) to (6,2.5);
    \draw [->,thick] (4,0.8) to (4,2.5);
\node at (4.5,0)  {$\psi(u)$};
\node at (6,0) {$\omega$};
   \node at (5.3,1) {$-u$};
   \node at (4.3,1) {$u$};
      \node at (4,2.8) {$c$};
    \node at (5,2.8) {$b$};
    \node at (6,2.8) {$a$};
  \end{tikzpicture}
  \caption{A pictorial interpretation of the ``square root relation''
    \eqref{gyokBYB}, which describes the exchange of the two-site
    block $\psi(u)$ and the one-site block $\omega$. This is a relation in
    $V_3\otimes V_2\otimes V_1\otimes End(V_A)$ and the outgoing indices
    $c,b,a$ describe the basis states in $V_3\otimes V_2\otimes
    V_1$. Fixing $c,b,a$ we obtain matrices acting on $V_A$. The local
    $\check R$ matrices acting at the crossings are defined such that 
  their argument is always the rapidity coming from the right minus the
  rapidity coming from the left. Thus we get an action of $\check
  R(u)$ on both sides, but on different vector spaces.}
  \label{fig:gyokBYBmps}
\end{figure}

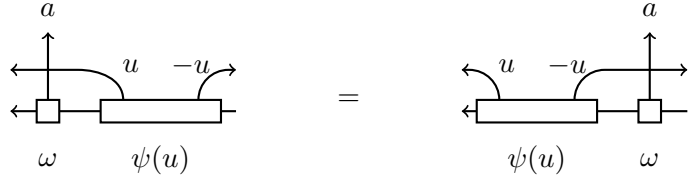
\begin{figure}
  \centering
  \begin{tikzpicture}
    \draw [<-,thick] (-0.5,0.65) to (-0.15,0.65);
    \draw [thick] (0.15,0.65) to (0.7,0.65);
    \draw [thick] (2.3,0.65) to (2.5,0.65);
      \draw [->,thick] (0,0.8) to  (0,1.7);
      \draw [thick] (-0.15,0.5) rectangle (0.15,0.8);
\draw [thick] (0.7,0.5) rectangle (2.3,0.8);
    \draw [->,thick] (1,0.8) to [out=90,in=0]  (0.4,1.2) to (-0.5,1.2);
    \draw [->,thick] (2,0.8) to [out=90,in=180] (2.4,1.2) to (2.5,1.2);
    \node at (0,0) {$\omega$};
    \node at (0,2) {$a$};
    \node at (1.5,0) {$\psi(u)$};
    \node at (1.9,1.25) {$-u$};
    \node at (1.1,1.25) {$u$};

\node at (4,0.8) {$=$};

\begin{scope}[xshift = 2 cm]
      \draw [<-,thick] (3.5,0.65) to (3.7,0.65);
    \draw [thick] (5.3,0.65) to (5.85,0.65);
    \draw [thick] (6.15,0.65) to (6.5,0.65);
      \draw [thick] (5.85,0.5) rectangle (6.15,0.8);
\draw [thick] (3.7,0.5) rectangle (5.3,0.8);
 \draw [->,thick] (6,0.8) to    (6,1.7);
    \draw [->,thick] (5,0.8) to [out=90,in=180]  (5.4,1.2) to (6.5,1.2);
    \draw [->,thick] (4,0.8) to [out=90,in=0] (3.6,1.2) to (3.5,1.2);
\node at (4.5,0)  {$\psi(u)$};
\node at (6,0) {$\omega$};
   \node at (6,2) {$a$};
   \node at (4.9,1.25) {$-u$};
   \node at (4.1,1.25) {$u$};
\end{scope}

  \end{tikzpicture}
  \caption{A pictorial interpretation of the ``square root relation''
    as an intertwining relation, see \eqref{intertwK}. Fixing the
    outgoing index $a$ we get a
    relation in $End(V_0\otimes V_A)$, where $V_0$ and $V_A$ are the
    auxiliary spaces of the Lax operators, and the $\omega_j$,
    respectively. The intertwining relation holds for all $a=1,\dots,N$.}
  \label{fig:intertw1}
\end{figure}

We also remark that from an alternative point of view, the sq.r.r. can
be used to define integrable MPS:
\begin{thm}
  If there is a solution to \eqref{gyokBYB} for some $\omega$,
  such that $K(u)$ is invertible for almost all $u$, 
  then the MPS built from $\omega$ is integrable.
\end{thm}
\begin{proof}
This follows immediately from the construction: the $K$-matrix
intertwines an arbitrary product of the dressed matrices $A_j$ and
$B_j$, and after multiplying with $K^{-1}(u)$ and taking the trace this leads to the
integrability condition \eqref{egyenloek2}.
\end{proof}
This proof of the integrability of the states is the one-site counterpart of our
earlier proof presented in \cite{sajat-integrable-quenches}, which only
treated two-site invariant cases constructed directly from the
$K$-matrices. This new proof can be used also in spin chains with odd
lengths, where the earlier method was not applicable. The
integrability of one-site states at odd lengths was first observed in \cite{sajat-minden-overlaps}.

If the intertwiner \eqref{intertwK} is unique, then it is group invariant:
\begin{equation}
  \label{group3}
  (\Lambda_0(G)\otimes \Lambda_\omega(G)) K(u)  =
  K(u)  (\Lambda_0(G^*)\otimes \Lambda_\omega(G)).
\end{equation}
This is easily seen by
contradiction: Assuming a non-invariant $K$-matrix it can be seen that
after the group transformation it also solves the same relation, and
by unicity it has to be proportional to the original $K(u)$. The
proportionality factor has to be a one-dimensional representation of
the group $G'$, and in our cases all such representations are
trivial. 

For the two site block
$\psi(u)$ the group invariance property takes the form
 \begin{equation}
  \label{group2}
  \Lambda_\omega(G^{-1}) \psi_{ab}(u)\Lambda_\omega(G)=G_{ac}G_{bd}  \psi_{cd}(u).
\end{equation}

It follows from the Yang-Baxter relations that $\tilde\psi(u)=\check
R(2u)\psi(-u)$ is also a solution to \eqref{BYB}, and \eqref{Rinit}
implies that it satisfies the same
initial condition $\tilde\psi_{ab}(0)=\psi_{ab}(0)=\omega_{a}\omega_b$. If the
solution of the sq.r.r. is unique, then we obtain the condition
\begin{equation}
  \label{psiRpsi}
 \psi(u)=f(u)\check R(2u)\psi(-u)
\end{equation}
with some function $f(u)$ satisfying $f(u)f(-u)=1$.

Assuming uniqueness, we can always re-define the normalization of $\psi(u)$ so that it satisfies
\begin{equation}
  \label{symrel}
 \psi(u)=\check R(2u)\psi(-u).
\end{equation}
This will be called the ``symmetry relation''. Note that this still
leaves room for an arbitrary re-definition 
$\tilde\psi(u)=g(u)\psi(u)$ with any even function $g(u)$ satisfying
$g(0)=1$. 

This symmetry relation is closely analogous to the ``boundary cross-unitarity relation'' in
integrable QFT (compare \eqref{symrel} to eq. (3.35) in \cite{Ghoshal:1993tm}).
  We note that relation \eqref{symrel} fixes the first derivative $\psi'(0)$ as
  \begin{equation}
    \psi'(0)=\check{R}'(0)\psi(0)=\check{R}'(0) (\omega\cdot \omega).
  \end{equation}

 The $R$-matrices are polynomials of the rapidity parameter (apart from
the normalization factor), therefore it
is natural to suspect that the relevant finite dimensional solutions to \eqref{gyokBYB} will
be polynomials as well.  Then the sq.r.r. has an immediate consequence for the asymptotic
behaviour of $\psi(u)$.
\begin{thm}
  For a given solution let $\alpha$ denote the highest degree of $u$
  in $\psi(u)$, and let $\phi_{ab}$ be the coefficient of
  $u^\alpha$. If the MPS is irreducible then all $\phi_{ab}$ are scalars.
  \end{thm}
\begin{proof}
  Let us take the sq.r.r. and take the $u\to\infty$ limit. The asymptotic value of the $\check R$ matrices is
  the permutation operator for both the $SU(N)$ and the $SO(N)$ case,
  therefore we obtain the simple commutativity condition
  \begin{equation}
    \phi_{ab} \omega_{c}=  \omega_{c}      \phi_{ab}.
  \end{equation}
If there are no irreducible subspaces, then it follows from Lemma
\eqref{irredU} that every $\phi_{ab}$ is
  proportional to the identity matrix in $V_A$.
\end{proof}

It can be seen that the limiting values $\phi_{ab}$ are either symmetric or
  anti-symmetric in the indices $a,b$.
  This follows most easily from relation \eqref{symrel}. We get a
  symmetric (or anti-symmetric) $\phi_{ab}$ if its degree $\alpha$ in
  $\psi_{ab}(u)$ is even (or odd), respectively.
  
In the simple case of the XXZ spin chain it was already shown in
\cite{sajat-integrable-quenches,sajat-minden-overlaps} that all one-site product states
(corresponding to an MPS with bond dimension 1) are integrable, and
they can be obtained from the well known $K$-matrices that solve the
usual Boundary Yang-Baxter relation. In Appendix \ref{sec:XXZ} we show
that the $K$-matrices associated to the one-site states indeed satisfy
the sq.r.r..

\subsection{The Boundary Yang-Baxter relation}

The twisted Boundary Yang-Baxter relation is
\begin{equation}
  \label{twisted}
  K_2(v) R^{T}_{21}(u+v)K_1(u)R_{12}(v-u)=
  R_{21}(v-u)K_1(u) R^T_{12}(u+v)K_2(v).
\end{equation}
Eq. \eqref{twisted} is sometimes also called the BYB relation for
soliton non-preserving boundary conditions
\cite{Delius-non-preserving,doikou-non-preserving}.
The main difference as opposed to the ``untwisted'' BYB relation is the
appearance of the transposition operation. Typically the
twisted BYB relation is formulated with different conventions: it might include
additional shifts in the rapidities (which can be compensated by a
redefinition of $K(u)$), or a different transposition defined as
\begin{equation}
  \label{furat}
  A^t=CA^TC
\end{equation}
with $C\in \text{End}(V_j), C^2=1$ being some crossing matrix. The
action of $C$ can be compensated by a basis transformation, and this
is discussed in Section \ref{sec:twisted}.
In the present work we use \eqref{twisted} and \eqref{BYB} because
these forms are most convenient to treat the MPS.

In those cases when the $R$-matrix itself satisfies a crossing relation
of the form
   \begin{equation}
     R^{T_1}(u)=C_1R(-u-\sigma)C_1
  \end{equation}
with some crossing matrix $C$ and crossing parameter
$\sigma\in\complex$, the eq. \eqref{twisted} is equivalent to the
standard BYB relation. In our 
examples the $SO(N)$-symmetric $R$-matrix is crossing invariant with
$C=1$ (see \eqref{SONcrossing}), but
the $SU(N)$-invariant is not. 

Making use of the identification \eqref{Kpsi} we can write the twisted
BYB in the  form
\begin{equation}
  \label{BYB}
  \check R_{12}(v-u)\check R_{23}(u+v)  (\psi(v)\cdot \psi(u))=
   \check R_{34}(v-u)\check R_{23}(u+v)  (\psi(u)\cdot \psi(v)),
 \end{equation}
 which is satisfied by the two-site block $\psi(u)$. This is a relation in $V_4\otimes V_3\otimes V_2\otimes V_1\otimes
 End (V_A)$, and it is understood that the $\check R$ matrices
 act on the respective components in the tensor product. 
For a graphical interpretation of \eqref{BYB} see Fig. \ref{fig:bybmps}.

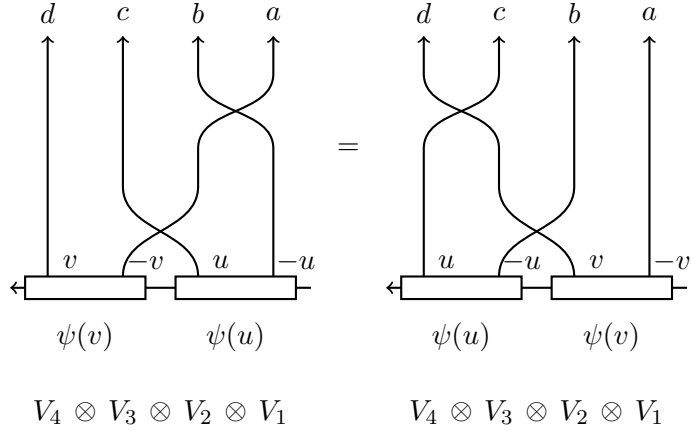
\begin{figure}
  \centering
  \begin{tikzpicture}

\draw [thick] (0.7,0.5) rectangle (2.3,0.8);
 \draw [thick] (2.7,0.5) rectangle (4.3,0.8);
 \draw [thick] (2.3,0.65) to (2.7,0.65);
 \draw [<-,thick] (0.5,0.65) to (0.7,0.65);
 \draw [thick] (4.3,0.65) to (4.5,0.65);
 
   [out=0,in=-90] (4,1);
   \draw [->,thick] (1,0.8) to (1,4);
   \draw [->,thick] (3,0.8) to [out=90,in=-90] (2,2) to (2,4);
   \draw [->,thick] (2,0.8) to [out=90,in=-90] (3,2) to (3,2.5) to
   [out=90,in=-90] (4,3.5) to (4,4);
   \draw [->,thick] (4,0.8) to (4,2.5) to [out=90,in=-90] (3,3.5) to (3,4);
   \node at (1.5,0) {$\psi(v)$};
     \node at (3.5,0) {$\psi(u)$};
   \node at (2.3,1) {$-v$};
   \node at (1.3,1) {$v$};
      \node at (3.3,1) {$u$};
      \node at (4.3,1) {$-u$};
      \node at (1,-1) {${V}_4$};
      \node at (2,-1) {${V}_3$};
      \node at (3,-1) {${V}_2$};
         \node at (4,-1) {${V}_1$};
         \node at (1.5,-1) {$\otimes$};
         \node at (2.5,-1) {$\otimes$};
         \node at (3.5,-1) {$\otimes$};
         \node at (1,4.3) {$d$};
         \node at (2,4.3) {$c$};
         \node at (3,4.3) {$b$};
         \node at (4,4.3) {$a$};

\node at (5,2.5) {$=$};
         
\draw [thick] (5.7,0.5) rectangle (7.3,0.8);
 \draw [thick] (7.7,0.5) rectangle (9.3,0.8);
 \draw [thick] (7.3,0.65) to (7.7,0.65);
 \draw [<-,thick] (5.5,0.65) to (5.7,0.65);
 \draw [thick] (9.3,0.65) to (9.5,0.65);
   \draw [->,thick] (9,0.8) to (9,4);
   \draw [->,thick] (8,0.8) to [out=90,in=-90] (7,2) to (7,2.5) to
   [out=90,in=-90] (6,3.5) to (6,4);
   \draw [->,thick] (7,0.8) to [out=90,in=-90] (8,2) to (8,4);
   \draw [->,thick] (6,0.8) to (6,2.5) to [out=90,in=-90] (7,3.5) to (7,4);
   \node at (6.5,0) {$\psi(u)$};
     \node at (8.5,0) {$\psi(v)$};
  \node at (7.3,1) {$-u$};
   \node at (6.3,1) {$u$};
      \node at (8.3,1) {$v$};
    \node at (9.3,1) {$-v$};

    \node at (6,-1) {${V}_4$};
      \node at (7,-1) {${V}_3$};
      \node at (8,-1) {${V}_2$};
         \node at (9,-1) {${V}_1$};
         \node at (6.5,-1) {$\otimes$};
         \node at (7.5,-1) {$\otimes$};
         \node at (8.5,-1) {$\otimes$};
             \node at (6,4.3) {$d$};
         \node at (7,4.3) {$c$};
         \node at (8,4.3) {$b$};
         \node at (9,4.3) {$a$};
  \end{tikzpicture}
  \caption{A pictorial representation for the BYB relation for the
    two-site MPS. $V_{1,2,3,4}$ denote the physical vector
    spaces, and $a,b,c,d$ are the physical indices. The matrices in
    the MPS act in the auxiliary space from the
    right to left. The local $\check R$ matrices acting at the crossings are defined such that
  their argument is always the rapidity coming from the right minus the
  rapidity coming from the left.}
  \label{fig:bybmps}
\end{figure}

The advantage of the representation \eqref{BYB} over \eqref{twisted} is that certain
 symmetry properties are more easily identified. In particular,  \eqref{BYB}
 involves the same $R$-matrix. This homogeneity in the exchange
 relation is the reason why it
 is always the twisted BYB which is relevant for integrable states of
 homogeneous spin chains.

In our earlier works
\cite{sajat-integrable-quenches,sajat-su3-1,sajat-su3-2} the BYB
relations were used as a starting point to define integrable initial
states. Here we take a different approach, and show that the BYB
relation actually follows from the integrability condition under
certain conditions. We intend to show that the $K$-matrix, which is
obtained as the intertwiner from \eqref{intertwK}-\eqref{Kpsi} indeed satisfies
the twisted BYB.
 
First of all we note that for the special point $u=0$ the
twisted BYB in the form \eqref{BYB} is equivalent to a double
application of the sq.r.r.. This follows from the initial condition
\eqref{init}. However, there is a more direct connection 
valid for arbitrary $u,v$.

Let us consider a double dressing of the MPS $\ket{\Psi_\omega}$ with
two transfer matrices. As it was already argued in
our previous work \cite{sajat-integrable-quenches} it follows from the integrability condition that
\begin{equation}
  t(u) t(v)\ket{\Psi_\omega}=\tilde t(u) \tilde t(v)\ket{\Psi_\omega}.
\end{equation}
Let us denote the auxiliary product space of the vectors above as $V_{a_1}\otimes V_{a_2}\otimes V_A$,
where $V_{a_1}$ and $V_{a_2}$ are the auxiliary spaces for the
transfer matrices $t(u)$ and $t(v)$, respectively.
We construct the corresponding two sets of dressed matrices as
\begin{equation}
   \begin{split}
    A_j&=\sum_{k,l} \mathcal{L}_{jk}(u)\otimes  \mathcal{L}_{kl}(v)
    \otimes \omega_l \\
    B_j&=\sum_{k,l} \mathcal{L}^T_{jk}(u)\otimes \mathcal{L}^T_{kl}(v)
    \otimes \omega_l. \\
  \end{split}
  \label{double}
\end{equation}

Our goal is to construct intertwiners for these
doubly dressed MPS:
\begin{equation}
  \label{doubleinter}
  M(u,v)A_j=B_j M(u,v).
\end{equation}
  We consider the following two candidates for the intertwiner
  $M(u,v)$:
  \begin{equation}
    \begin{split}
      M_1(u,v)&=   K_2(v) R^{T}_{21}(u+v)K_1(u)R_{12}(v-u)\\
      M_2(u,v)&=  R_{21}(v-u)K_1(u) R^T_{12}(u+v)K_2(v).
    \end{split}
  \end{equation}
Here it is understood that $K_{1,2}(u)$ act on $V_{a_{1,2}}$, respectively. For a
graphical interpretation of these two intertwiners see
Fig. \ref{fig:BYBintertw}.

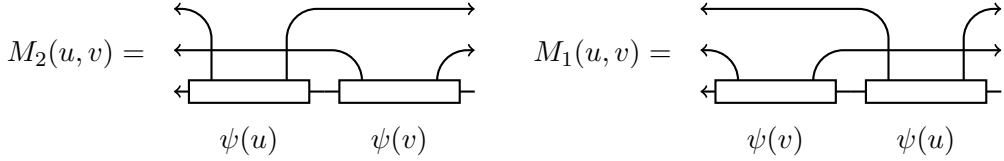
\begin{figure}
  \centering
  \begin{tikzpicture}
    \begin{scope}[xshift=-7cm]
      \node at (2.2,1.5) {$M_2(u,v)=$};
      \pic at (3,0) {psiny};
      \node at (6.5,0.35) {$\psi(v)$};
          \node at (4.5,0.35) {$\psi(u)$};
      \pic at (5,0) {psi};
        \draw [->,thick] (4,1.15) to (4,1.7) to [out=90,in=0]  (3.6,2.1) to (3.5,2.1);
        \draw [->,thick] (5,1.15) to (5,1.7) to [out=90,in=180] (5.4,2.1) to (7.5,2.1);
 \draw [->,thick] (6,1.15) to [out=90,in=0]    (5.6,1.55) to (3.5,1.55);
 \draw [->,thick] (7,1.15)  to [out=90,in=180] (7.4,1.55) to (7.5,1.55);
\end{scope}

   \node at (2.2,1.5) {$M_1(u,v)=$};
    \pic at (3,0) {psiny};
      \node at (4.5,0.35) {$\psi(v)$};
          \node at (6.5,0.35) {$\psi(u)$};
      \pic at (5,0) {psi};
        \draw [->,thick] (4,1.15) to [out=90,in=0]  (3.6,1.55) to (3.5,1.55);
        \draw [->,thick] (5,1.15) to [out=90,in=180] (5.4,1.55) to (7.5,1.55);
 \draw [->,thick] (6,1.15) to (6,1.7) to [out=90,in=0]    (5.6,2.1) to (3.5,2.1);
 \draw [->,thick] (7,1.15) to (7,1.7) to [out=90,in=180] (7.4,2.1) to
 (7.5,2.1);
  \end{tikzpicture}
  \caption{The two intertwiners which are later identified as the two
    sides of the twisted BYB. Note that the order of $\psi(u)$ and
    $\psi(v)$ is reversed, but the upper (lower) lines are always
    associated to the rapidities $u$ and $v$, respectively.}
  \label{fig:BYBintertw}
\end{figure}

 Multiple use of the Yang-Baxter relation and the simple intertwiner
 relation \eqref{intertwK} shows that both $M_1(u,v)$ and $M_2(u,v)$
 satisfy \eqref{doubleinter}. The idea is that $K_1(u)$ always
 intertwines the dressing with $\mathcal{L}(u)$, and $K_2(v)$
 intertwines $\mathcal{L}(v)$, and the order of the exchange with
 $K_1(u)$ and $K_2(v)$ is just the opposite for $M_1(u,v)$ and
 $M_2(u,v)$. We also use that in the intermediate steps the 
 order of the dressings of the MPS can be exchanged as well, using the
 standard RTT relation. For a graphical interpretation of the steps of
 the proof see \ref{fig:doubleintertw}.
 
\begin{figure}
  \centering
  \begin{tikzpicture}

    \def\ymax{2.7}

    \foreach \x [count=\n]in {1,2,3}{
      \pic at (\x,0) {omega};
    }
      \pic at (0,0) {omegany};
      \pic at (3,0) {psi};
      \node at (4.5,0.35) {$\psi(v)$};
          \node at (6.5,0.35) {$\psi(u)$};
      \pic at (5,0) {psi};
        \draw [->,thick] (4,1.15) to [out=90,in=0]  (3.6,1.55) to (-0.5,1.55);
        \draw [->,thick] (5,1.15) to [out=90,in=180] (5.4,1.55) to (7.5,1.55);
 \draw [->,thick] (6,1.15) to (6,1.7) to [out=90,in=0]    (5.6,2.1) to (-0.5,2.1);
 \draw [->,thick] (7,1.15) to (7,1.7) to [out=90,in=180] (7.4,2.1) to (7.5,2.1);

 \begin{scope}[yshift=-4cm]
   \foreach \x [count=\n]in {2,3,4,5}{
      \pic at (\x,0) {omega};
}
      \pic at (-1,0) {psiny};
      \node at (0.5,0.35) {$\psi(v)$};
          \node at (6.5,0.35) {$\psi(u)$};
          \pic at (5,0) {psi};
       \draw [->,thick] (0,1.15) to [out=90,in=0]  (-0.4,1.55) to (-0.5,1.55);
        \draw [->,thick] (1,1.15) to [out=90,in=180] (1.4,1.55) to (7.5,1.55);
 \draw [->,thick] (6,1.15) to (6,1.7) to [out=90,in=0]    (5.6,2.1) to (-0.5,2.1);
 \draw [->,thick] (7,1.15) to (7,1.7) to [out=90,in=180] (7.4,2.1) to (7.5,2.1);
 \end{scope}

 \begin{scope}[yshift=-8cm]
   \foreach \x [count=\n]in {2,3,4,5}{
      \pic at (\x,0) {omega};
}
      \pic at (-1,0) {psiny};
      \node at (0.5,0.35) {$\psi(v)$};
          \node at (6.5,0.35) {$\psi(u)$};
          \pic at (5,0) {psi};
       \draw [->,thick] (0,1.15) to [out=90,in=0]  (-0.4,1.55) to (-0.5,1.55);
        \draw [->,thick] (1,1.15) to (1,1.7) to [out=90,in=180]
        (1.4,2.1) to (6.5,2.1) to [out=0,in=180] (7.5,1.55);
 \draw [->,thick] (6,1.15) to [out=90,in=0]    (5.6,1.55) to
 (0.5,1.55) to [out=180,in=0] (-0.5,2.1);
 \draw [->,thick] (7,1.15) to (7,1.7) to [out=90,in=180] (7.4,2.1) to (7.5,2.1);
 \end{scope}

 \begin{scope}[yshift=-12cm]
   \foreach \x [count=\n]in {4,5,6,7}{
      \pic at (\x,0) {omega};
}
      \pic at (-1,0) {psiny};
      \node at (0.5,0.35) {$\psi(v)$};
          \node at (2.5,0.35) {$\psi(u)$};
          \pic at (1,0) {psi};
       \draw [->,thick] (0,1.15) to [out=90,in=0]  (-0.4,1.55) to (-0.5,1.55);
        \draw [->,thick] (1,1.15) to (1,1.7) to [out=90,in=180]
        (1.4,2.1) to (7.5,2.1) to [out=0,in=180] (8.5,1.55);
 \draw [->,thick] (2,1.15) to [out=90,in=0]    (1.6,1.55) to
 (0.5,1.55) to [out=180,in=0] (-0.5,2.1);
 \draw [->,thick] (3,1.15) to [out=90,in=180] (3.4,1.55) to (7.5,1.55)
 to [out=0,in=180] (8.5,2.1);
\end{scope}

 \begin{scope}[yshift=-16cm]
   \foreach \x [count=\n]in {4,5,6,7}{
      \pic at (\x,0) {omega};
}
      \pic at (-1,0) {psiny};
      \node at (0.5,0.35) {$\psi(v)$};
          \node at (2.5,0.35) {$\psi(u)$};
          \pic at (1,0) {psi};
       \draw [->,thick] (0,1.15) to [out=90,in=0]  (-0.4,1.55) to (-0.5,1.55);
        \draw [->,thick] (1,1.15) to [out=90,in=180]  (1.4,1.55) to (7.5,1.55);
 \draw [->,thick] (2,1.15) to (2,1.7) to [out=90,in=0]    (1.6,2.1) to
 (-0.5,2.1);
 \draw [->,thick] (3,1.15) to (3,1.7) to [out=90,in=180] (3.4,2.1) to (7.5,2.1);
 \end{scope}

   \end{tikzpicture}
   \caption{A graphical demonstration for the double intertwining
     using $M_1(u,v)$, which is the object shown on the r.h.s. of
     Fig. \ref{fig:BYBintertw}. Here we intertwine multiple products of the
     doubly dressed MPS, by multiple use of the simple intertwining
     relation \eqref{intertwK} (see Fig. \ref{fig:intertw1}) and the
     Yang-Baxter relations \eqref{YB0}.
     Essentially the same steps (although in
    a different order) can be repeated
     also with $M_2(u,v)$, which is shown on the l.h.s. of Fig.
     \ref{fig:BYBintertw}. If the intertwiner is unique
     then $M_1$ and $M_2$ have to be proportional to each other. This
     intertwining relation thus establishes a connection between the
     twisted BYB and the integrability condition.
}
  \label{fig:doubleintertw}
\end{figure}
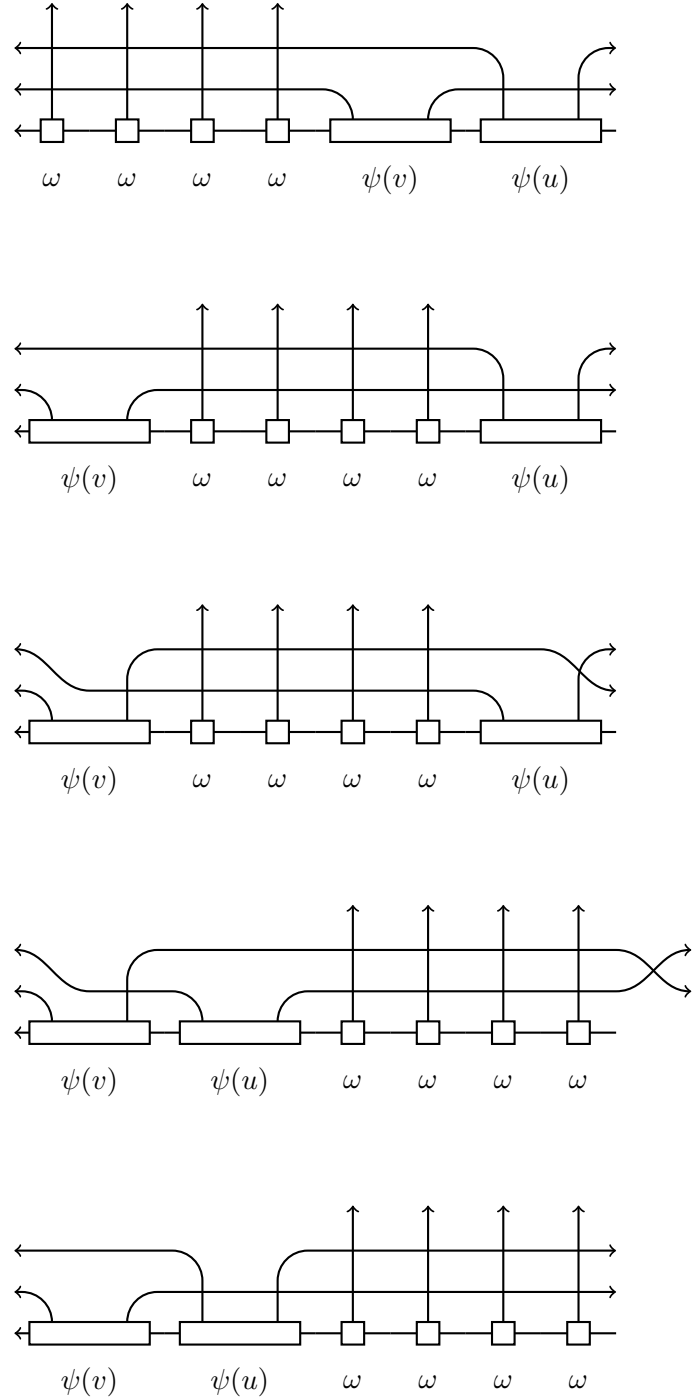

If the doubly dressed MPS are irreducible then
Theorem
\ref{simmm} states that the intertwiners are unique up to
normalization, which means that
\begin{equation}
  M_1(u,v)=g(u,v)M_2(u,v)
\end{equation}
with some function $g(u,v)$.
Investigating specific components of the relation
\eqref{BYB} it can be seen that $g(u,v)\equiv 1$. Thus, given the
irreducibility property we have established that the solution of the sq.r.r. also
solves the twisted BYB relation \eqref{twisted}. 

Unfortunately we have not yet managed to prove the twisted BYB
independently from the irreducibility property.
In our concrete examples we have checked numerically that the doubly
dressed MPS are indeed irreducible for generic $u,v$ (see Appendix \ref{sec:C1}), but
it would be desirable to establish this by purely analytic means.

Of course it can be checked by direct substitution whether a specific
solution to \eqref{gyokBYB} also solves \eqref{BYB}. We performed this
in our concrete examples (presented later in Sec. \ref{sec:su} and
\ref{sec:so}) and found agreement. 

\subsection{Implications of integrability: the Quantum Transfer Matrix}

Here we show that the BYB relation \eqref{BYB} allows for the
construction of commuting sets of double row transfer matrices.
The construction is
essentially the same as in the papers dealing with soliton
non-preserving boundary conditions \cite{doikou-non-preserving,ragoucy-doikou-non-preserving}. 
In the present context these double row objects 
will be called Quantum Transfer Matrices (QTM's), in analogy with the
thermal case \cite{kluemper-review}.

Instead of the homogeneous chain it is useful to introduce an alternating sequence of
inhomogeneities $(-u_1,u_1,-u_2,u_2,\dots,-u_{L/2},u_{L/2})$. The
parameters $u_j$ will play the role of spectral parameters for
``double-row transfer matrices'' to be constructed.

We define two inhomogeneous transfer matrices as
\begin{equation}
  \label{inhomT}
  \begin{split}
    t(v|u_1,\dots,u_{L/2})&=\text{Tr}\ T(v|u_1,\dots,u_{L/2}), \\
    T(v|u_1,\dots,u_{L/2})&=R_{0L}(v-u_{L/2})R_{0,L-1}(v+u_{L/2})\dots R_{02}(v-u_1)R_{01}(v+u_1)\\
    \tilde t(v|u_1,\dots,u_{L/2})&=\text{Tr}\  \tilde T(v|u_1,\dots,u_{L/2}),\\
  \tilde T(v|u_1,\dots,u_{L/2})&= R_{01}(v-u_1) R_{02}(v+u_1) \dots R_{0,L-1}(v-u_{L/2})R_{0L}(v+u_{L/2}).\\
  \end{split}
\end{equation}
Here the $L$ inhomogeneities are of alternating signs, and in the notation
we write them as $(u_1,\dots,u_{L/2})$.

We also define inhomogeneous initial states as
\begin{equation}
  \label{inhomMPS}
    |\Psi(u_1,u_2,\dots,u_{L/2})\rangle=\sum_{i_1,\dots,i_L=1}^N{\rm tr}_{0}
    \left[
\psi_{i_{L},i_{L-1}}(u_{L/2})\dots \psi_{i_2,i_1}(u_1)
\right]
|i_L,\dots,i_1\rangle.
\end{equation}
The physical states \eqref{psidef} are reproduced in the homogeneous
limit $u_j\to 0$ due to the initial condition \eqref{init}.

Using the unitarity condition we can write \eqref{BYB} in the form
\begin{equation}
  \label{BYB2}
  \psi(v)\cdot\psi(u)=
 \check R_{23}(-u-v)  \check R_{12}(u-v)  \check R_{34}(v-u)\check
 R_{23}(u+v)  (\psi(u)\cdot \psi(v)).
\end{equation}
This exchange relation can be extended to the full MPS,
for example
\begin{equation}
  \begin{split}
&  |\Psi(u_1,u_2,\dots,u_{L/2})\rangle=\\
 &  \check R_{23}(-u_1-u_2)  \check R_{12}(u_1-u_2)  \check R_{34}(u_2-u_1)\check R_{23}(u_1+u_2)
    |\Psi(u_2,u_1,\dots,u_{L/2})\rangle,
\end{split}
  \end{equation}
and similarly for other exchanges of the inhomogeneity
parameters. It follows from the Yang-Baxter equations that this exchange relation is also compatible with the
inhomogeneous monodromy matrices, and we have for example
\begin{equation}
  \begin{split}
&t(v|u_1,u_2,\dots,u_{L/2})  |\Psi(u_1,u_2,\dots,u_{L/2})\rangle=
\check R_{23}(-u_1-u_2)  \check R_{12}(u_1-u_2) \times \\
&\hspace{0.5cm}\times \check R_{34}(u_2-u_1)\check R_{23}(u_1+u_2)
   \Big[ t(v|u_2,u_1,\dots,u_{L/2})   |\Psi(u_2,u_1,\dots,u_{L/2})\rangle\Big].
   \end{split}
\end{equation}
Similar relations hold for arbitrary products of the transfer
matrices\footnote{Such exchange relations hold for any transfer 
  matrix which is built from Lax operators satisfying the fundamental exchange
relation dictated by the $R$-matrix. They are the so-called ``fused
transfer matrices'', and they also include the space reflected
fundamental transfer matrix $\tilde t(v|u_1,u_2,\dots)$, see for
example \cite{sajat-su3-1,sajat-su3-2}. In this work
we don't discuss the fusion hierarchy in detail, therefore the main discussion
is limited to the products of the fundamental transfer matrix.}.

Let us consider partition functions involving two
different MPS, that serve as initial and final states. They will be denoted as $\ket{\Psi_A}$ and
$\bra{\Psi_B}$, and the different subscripts indicate that they are
not necessarily adjoints of each other. It is important that both
two-site blocks $\psi_{A}(u)$ and $\psi_B(u)$ satisfy the same twisted
reflection relation \eqref{BYB}.

We define the inhomogeneous dual MPS vectors as
\begin{equation}
  \label{inhomMPSb}
    \bra{\Psi_B(u_1,u_2,\dots,u_{L/2})}=\sum_{i_1,\dots,i_L=1}^N{\rm tr}_{0}
    \left[
\psi_{B,i_{L},i_{L-1}}(-u_{L/2})\dots \psi_{B,i_2,i_1}(-u_1)
\right]
\bra{i_L,\dots,i_1}.
\end{equation}
It is important that the rapidity parameters are taken with a sign difference. As an effect these states
satisfy the dual exchange relation
\begin{equation}
  \begin{split}
&\bra{\Psi_B(u_1,u_2,u_3,\dots,u_{L/2})}
 \check R_{23}(-u_1-u_2)  \check R_{12}(u_1-u_2)  \check R_{34}(u_2-u_1)\check R_{23}(u_1+u_2)
  =\\
 & 
 \bra{\Psi_B(u_2,u_1,u_3,\dots,u_{L/2})},
 \end{split}
\end{equation}
and similarly for exchanges of other rapidity pairs.

Let us consider the partition functions 
\begin{equation}
  \label{ZAB1}
  \begin{split}
&  Z_{AB}(v_1,\dots,v_m|u_1,\dots,u_{L/2})=\\
& \hspace{2cm} \bra{\Psi_B(u_1,\dots,u_{L/2})}\prod_{j=1}^m
  t(v_j|u_1,\dots,u_{L/2})   \ket{\Psi_A(u_1,\dots,u_{L/2})}.
  \end{split}
\end{equation}
The $Z_{AB}$ are completely symmetric in both the $u$-
and the $v$-parameters. Symmetry with respect to $v_j$, $j=1\dots m$
follows from the commutativity of the transfer matrices, whereas
symmetry with respect to $u_j$, $j=1\dots L/2$ follows from the
above exchange relations involving the initial and final states. For a
pictorial interpretation of the partition functions see Fig. \ref{fig:ZAB}.

In the physical applications it is usually required that the final
state is the dual vector to the initial state, which results in the
requirement
\begin{equation}
  \psi_{B,j,k}(0)=(\psi_{A,j,k}(0))^*,\qquad j,k=1\dots N.
\end{equation}
Generally this means that $\psi_A(u)$ and $\psi_B(u)$ are two distinct
solutions to \eqref{BYB}, unless all the matrices can be chosen as
completely real.

The $Z_{AB}$ can be interpreted as the Loschmidt amplitude for some discrete
time evolution: in the homogeneous limit $u_j\to 0$ they can
be used to approximate the physical Loschmidt amplitude
\begin{equation}
  Z_{AB}\approx \bra{\Psi_0}e^{-iHt}\ket{\Psi_0}.
\end{equation}
For the details of this procedure we refer the reader to
\cite{sajat-integrable-quenches}. 

The partition functions \eqref{ZAB1} allow for an alternative evaluation, which leads
to the introduction of the double row Quantum Transfer Matrices, which
act in the horizontal direction in Fig. \ref{fig:ZAB}.
It can be read off
Fig. \ref{fig:ZAB} (or it can be established by purely algebraic
means) that their explicit form is
\begin{equation}
  \label{TABdef}
  \begin{split}
 &   \mathcal{T}_{AB}(u|v_1,\dots,v_m)=\\
   & \sum_{a_1,a_2,b_1,b_2=1}^N
  \psi_{A,b_2,b_1}(u)\otimes \Big[
  T_{a_2b_2}(-u|-v_1,\dots,-v_m)T_{a_1b_1}(u|-v_1,\dots,-v_m)\Big]
  \otimes  \psi_{B,a_2,a_1}(-u).
  \end{split}
\end{equation}
Alternatively this can be computed as
\begin{equation}
  \mathcal{T}_{AB}(u)=\text{Tr} \left(M_{A}(u)K^T_B(-u)\right),
\end{equation}
where
\begin{equation}
  M_A(u)=T(-u)K_A(u)T^T(u)
\end{equation}
is the ``quantum monodromy matrix'',
and the products and traces above
are to be understood in the indices $a_{1,2}$, $b_{1,2}$ which
originally label
the states of the physical spin chain. 

The partition function is then evaluated as
\begin{equation}
    Z_{AB}(v_1,\dots,v_m|u_1,\dots,u_{L/2})=\text{Tr}\left[
\prod_{j=1}^{L/2}  \mathcal{T}_{AB}(u_j|v_1,\dots,v_m)\right].
  \end{equation}
The symmetry properties of $Z_{AB}$ are equivalent to the
commutativity condition
\begin{equation}
  [  \mathcal{T}_{AB}(u_1),  \mathcal{T}_{AB}(u_2)]=0,
\end{equation}
which can be proven directly using \eqref{BYB} and the Yang-Baxter relations.

We note that depending on the circumstances the double row transfer
matrices \eqref{TABdef} can be used to define integrable boundary
models with local Hamiltonians and additional degrees of freedom at
the two boundaries. For an example of such a model see
\cite{dewolfe-mann}. In Sections \ref{sec:su} and \ref{sec:so} we 
present explicit examples for the integrable two site block
$\psi_{ab}(u)$, which are new in some cases and thus would lead to new
integrable boundary conditions. However, it is not the purpose of the
present paper to investigate the Hamiltonians and/or the spectra of
these models, and this is left to further work.

\begin{figure}
  \centering
  \begin{tikzpicture}[scale=1]
\node [white] at (0,7) {.};
    \foreach \y [count=\n]in {0,5}{
        \begin{scope}[yshift = \y cm]
\draw [thick] (0.7,0.5) rectangle (2.3,0.8);
\draw [thick] (2.7,0.5) rectangle (4.3,0.8);
 \draw [thick] (4.7,0.5) rectangle (6.3,0.8);
 \draw [thick] (2.3,0.65) to (2.7,0.65);
 \draw [<-,thick] (0.5,0.65) to (0.7,0.65);
 \draw [thick] (4.3,0.65) to (4.7,0.65);
\draw [thick] (6.3,0.65) to (6.5,0.65);
      \end{scope}
    }

       \foreach \y [count=\n]in {1,2,3,4}{
         \begin{scope}[yshift = \y cm]
           \draw [<-,thick] (0.5,0.65) to (6.5,0.65);
   \end{scope}
    }
          
\foreach \x [count=\n]in {0,2,4}{ 
   \begin{scope}[xshift = \x cm]
 \draw [->,thick] (1,0.8) to (1,5.5);
 \draw [->,thick] (2,0.8) to (2,5.5);
          \end{scope}
    }

    \node at (7,1.65) {$v_1$};
    \node at (7,2.65) {$v_2$};
    \node at (7,3.65) {$v_3$};
       \node at (7,4.65) {$v_4$};
   \node at (1.5,0) {$\psi_A(u_1)$};
   \node at (3.5,0) {$\psi_A(u_2)$};
   \node at (5.5,0) {$\psi_A(u_3)$};
      \node at (1.5,6.3) {$\psi_B(-u_1)$};
   \node at (3.5,6.3) {$\psi_B(-u_2)$};
      \node at (5.5,6.3) {$\psi_B(-u_3)$};
   \node at (1.3,1) {$u_1$};
   \node at (2.3,1) {$-u_1$};
      \node at (3.3,1) {$u_2$};
      \node at (4.3,1) {$-u_2$};
         \node at (5.3,1) {$u_3$};
      \node at (6.3,1) {$-u_3$};

    \end{tikzpicture}
    \caption{An example for a partition function with integrable
      boundaries. In the original physical picture the bottom and top
      rows are interpreted as the MPS which serve as the initial and final states for some
      discrete time evolution. Alternatively, the partition function
      can be evaluated by the Quantum Transfer Matrix, which acts in
      the horizontal direction: in this picture the two-site blocks
      play the role of integrable boundaries, with an additional
      degree of freedom at the boundary.}
  \label{fig:ZAB}
\end{figure}
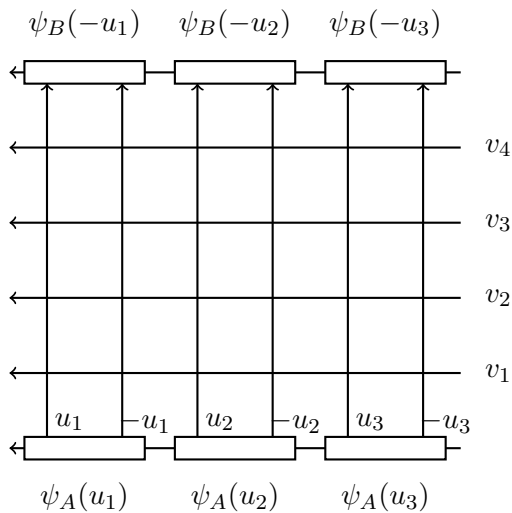

\section{Relation with the twisted Yangian}

\label{sec:twisted}

In the following we describe the twisted Yangians that are relevant to
the integrable MPS. Our goal is to establish a direct relation to
the abstract algebraic setting; in this Section we limit ourselves to the case of the
symmetric pair $(SU(N),SO(N))$.  

The twisted Yangian $Y^+(N)$ corresponding
to the orthogonal Lie algebra $\mathfrak{o}(N)$ is an abstract algebra
with generators $s^{(j)}_{kl}$, $j=1,2,\dots,\infty$, $k,l=1,\dots,N$ given by the
following exchange relations \cite{olshanskii2,molev-yangians-review}. Let us define the matrix
$S(u)$ as a formal series in $u^{-1}$ as
\begin{equation}
  \label{Sdef}
  S(u)=\sum_{k,l=1}^N s_{kl}(u)\otimes E_{kl},\qquad
  s_{kl}(u)=\delta_{kl}+\sum_j s^{(j)}_{kl}u^{-j}.
\end{equation}
Then the exchange relations are given by the formal equations
\begin{equation}
  \begin{split}
  \label{twistedY}
  R(v-u)  S_2(v) R^{t}(-u-v)S_1(u)&=
S_1(u) R^t(-u-v)S_2(v)R(v-u)\\
    S^t(-u)&=S(u)+\frac{S(u)-S(-u)}{2u}.
  \end{split}
\end{equation}
Here the transposition is defined as in \eqref{furat} with a charge
conjugation matrix given by $C_{j,k}=\delta_{j,N-k+1}$. The second
relation is equivalent to our symmetry relation \eqref{symrel} 
specified to the $SU(N)$-symmetric $R$-matrix.

The conventions of the relations \eqref{twistedY} are different from the
ones used in our 
\eqref{twisted}.
The first difference is simply a sign convention for the rapidity
parameter, which originates from the convention $R(u)=1-u^{-1}P$ used in
\cite{molev-representations}. This can be compensated immediately.
The second difference is that our equations involve only
standard transpositions.  Nevertheless the two sets of equations can be transformed into each other
by a basis transformation: let $Z\in SU(N)$ be a matrix
satisfying
\begin{equation}
  ZZ^T=C.
\end{equation}
Such a matrix is easily found: in the
case of $N=2$ we can choose
\begin{equation}
  Z=  \frac{1}{\sqrt{2}}
  \begin{pmatrix}
1  & i  \\
  1   & -i   \\
  \end{pmatrix},
\end{equation}
and for higher $N$ we can build $Z$ analogously using the block structure
of $C$. 
If $S(u)$ satisfies the defining relations \eqref{twistedY} then
the transformed matrix defined as
\begin{equation}
  \label{KS}
  K(u)=Z^{-1}S(-u)Z
\end{equation}
satisfies our \eqref{twisted}.
This can be checked by direct computation, using the group invariance
properties \eqref{Rgroup}-\eqref{Rgroup2} and also $Z^{-1}=Z^TC$ which
follows from $C^2=1$.

The definition
\eqref{Sdef} itself can be interpreted as an asymptotic condition: it
is equivalent to the requirement that in $S_{ab}(u)$ (or $\psi_{ab}(u)$) the leading term in
$u$ should be $\delta_{ab}$. This requirement is invariant with respect to the
basis transformation \eqref{KS}.

The finite dimensional irreducible representations of $Y^+(N)$ have been characterized in
\cite{molev-representations}. Any such representation gives a concrete
solution $\psi_{ab}(u)$ to \eqref{BYB}, which will result in an
integrable two-site invariant MPS. We will be looking for the subset
of solutions which also satisfy the factorizability condition
\eqref{init}. One possible strategy could be to
survey all possible irreps and to investigate the possibility of a
factorization separately. Instead, we use a different method: we construct concrete solutions to a
given $\omega$  by solving the sq.r.r., this is presented in the next
two Sections. Nevertheless
here we review the main statements of \cite{molev-representations},
translated into our framework through relations \eqref{KS} and \eqref{Kpsi}.

First we define Lax-operators that generate the finite dimensional
representations of the Yangian. Let $\Lambda$ be a finite dimensional
irreducible representation of $\mathfrak{sl}(N)$ acting on the vector
space $V_\Lambda$.
For the representation of the $E_{ab}$ standard basis matrices we will
use the notation $E_{ab}^\Lambda$.
We define the Lax operator $ \mathcal{L}^{(1,\Lambda)}(u)$ acting
on the product of a physical and auxiliary space $V_1\otimes V_\Lambda$ as
\begin{equation}
  \mathcal{L}^{(1,\Lambda)}_{ab}(u)=u\delta_{ab}+E^\Lambda_{ba},
\end{equation}
where $a,b=1\dots N$ are interpreted as the physical indices, and the
$E^\Lambda_{ab}$ act on the auxiliary space.

From any such Lax operator we can build transfer matrices in the usual
way. In the inhomogeneous case discussed above
\begin{equation}
  t^{(\Lambda)}(v)=\text{Tr}_\Lambda\Big[
   \mathcal{L}^{(L,\Lambda)}(v-u_{L/2})
   \mathcal{L}^{(L-1,\Lambda)}(v+u_{L/2})\dots
   \mathcal{L}^{(2,\Lambda)}(v-u_1)
   \mathcal{L}^{(1,\Lambda)}(v+u_1)\Big]. \\
\end{equation}
These transfer matrices commute even for different representations:
\begin{equation}
  [t^{(\Lambda)}(v),t^{(\Lambda')}(v')]=0.
\end{equation}

Let us denote by $F_{ab}=E_{ab}-E_{ab}^t$ the generators of
$\mathfrak{o}(N)$ in the defining representation. Let us take an irreducible representation $\Omega$ of 
$\mathfrak{o}(N)$ with the matrices $F_{ab}^\Omega$ acting on the
vector space
$V_\Omega$.
Then we can construct the irreducible solution
$\psi^\Omega(u)$ 
to the BYB with auxiliary space $V_\Omega$ as \cite{molev-representations}
\begin{equation}
  \label{rootsolutions}
  \psi^\Omega_{ab}(u)=\delta_{ab}+F^\Omega_{ab}(u+1/2)^{-1}.
\end{equation}
Such solutions can be called ``root solutions'' as they serve as
starting points to construct all finite dimensional representations of
the twisted Yangian. They obviously satisfy the group symmetry
requirement \eqref{group2} with $\mathcal{G}'=SO(N)$. The MPS
corresponding to a root solution $\psi^\Omega(u)$ will be denoted by
$\ket{\Psi^\Omega}$.

Further non-trivial solutions can be obtained by ``dressing'' these
root solutions with the Lax-operators above. Taking any solution
$\psi(u)$ a new solution is constructed as
\begin{equation}
  \label{dressedpsi}
\tilde \psi_{ab}(u)= \mathcal{L}_{ac}^{(\Lambda)}(v-u) \mathcal{L}_{bd}^{(\Lambda)}(v+u) \psi_{cd}(u).
\end{equation}
In terms of the $K$-matrices this dressing is given by
\begin{equation}
  \label{dressedK}
  \tilde K(u)= \mathcal{L}^{(0,\Lambda)}(v-u)   K(u) (\mathcal{L}^{(0,\Lambda)}(v+u))^t.
\end{equation}
where $0$ denotes the physical space for the action of the
$K$-matrix.  In terms of the algebraic setting the dressing \eqref{dressedK}
describes the co-ideal property of the twisted Yangian within the
original Yangian with the usual co-product. For a graphical
interpretation of this dressing see the first figure of \ref{fig:dressing}.

\begin{figure}
  \centering
  \begin{tikzpicture}
    \draw [<-,thick] (0.5,0.65) to (0.7,0.65);
    \draw [thick] (2.3,0.65) to (2.5,0.65);
    \draw [->,thick] (1,0.8) to (1,1.7);
    \draw [->,thick] (2,0.8) to (2,1.7);
    \draw [<-,thick,dashed] (0.5,1.25) to (2.5,1.25);
       \draw [thick] (0.7,0.5) rectangle (2.3,0.8);
\node at (1.5,0.1) {$\psi(u)$};
 \node at (1.3,1.8) {$u$};
   \node at (2.3,1.8) {$-u$};       
   \node at (2.7,1.25) {$v$};

   \begin{scope}[xshift=6cm]
         \draw [thick] (0.5,0.65) to (0.7,0.65);
    \draw [thick] (2.3,0.65) to (2.5,0.65);
    \draw [->,thick] (1,0.8) to (1,2.4);
    \draw [->,thick] (2,0.8) to (2,2.4);
    \draw [<-,thick,dashed] (0.5,1.25) to (2.5,1.25);
       \draw [<-,thick,dashed] (0.5,1.75) to (2.5,1.75);
       \draw [thick] (0.7,0.5) rectangle (2.3,0.8);
\node at (1.5,0.1) {$\psi(u)$};
 \node at (1.3,2.5) {$u$};
   \node at (2.3,2.5) {$-u$};       

 \draw [thick] (0.3,0.5) rectangle (0.5,2);
 \draw [->,thick] (0.3,1.25) to (0,1.25);
 \node at (0.4,0.1) {$\Pi$};

 \draw [thick] (2.5,0.5) rectangle (2.7,2);
 \draw [<-,thick] (2.7,1.25) to (3,1.25);
 \node at (2.6,0.1) {$\Pi$};

\end{scope}

  \end{tikzpicture}
  \caption{Examples for dressed solutions. The first example involves
    an ``LKL''-type dressing, which on the level of MPS corresponds to
  the action with a transfer matrix. Generally we can perform multiple
  instances of these dressings, but the resulting MPS might
  have invariant subspaces. Irreducible representations of the twisted
  Yangian are then obtained using a projector $\Pi$ onto an
  irreducible sup-space in the tensor product of auxiliary spaces.}
  \label{fig:dressing}
\end{figure}
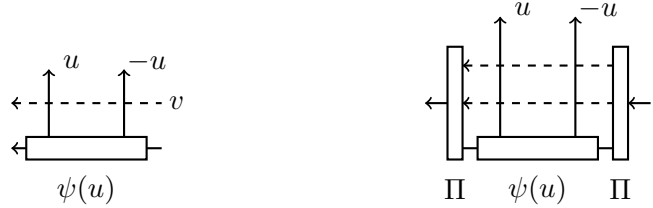

It is immediately seen that the MPS obtained from the dressed two-site
block \eqref{dressedpsi} will be
\begin{equation}
  \label{dressedMPS}
  \ket{\tilde \Psi}=t^{(\Lambda)}(v) \ket{\Psi}.
\end{equation}
This idea to construct new integrable MPS by the action of
transfer matrices already appeared in
\cite{zarembo-neel-cite3,sajat-integrable-quenches}. The commutativity of the transfer
matrices can be used to obtain an immediate proof of the integrability of the ``dressed'' MPS.

In principle any finite number of dressings can be performed leading
to
\begin{equation}
  \label{dressedK2}
   \tilde K(u)=\left(\prod_{j=1}^n  \mathcal{L}^{(0,\Lambda_j)}(v_j-u)\right)
   K(u) \left(\prod_{j=n}^1 (\mathcal{L}^{(0,\Lambda_j)}(v_j+u))^t\right)
\end{equation}
and MPS' of the form
\begin{equation}
\prod_{j=1}^n  t^{(\Lambda_j)}(v_j) \ket{\Psi^\Omega}.
\end{equation}
However, these MPS are not necessarily irreducible, even if we used
irreducible representations $\Omega$ and $\Lambda_j$ of
$\mathfrak{o}(N)$ and $\mathfrak{gl}(N)$. A well known example for
reducibility is when the product of transfer matrices has invariant
subspaces in the tensor product of auxiliary spaces: this happens
during the so-called fusion procedure. However, this is not the only
possibility: the dressed two-site blocks can have invariant subspaces
even if a single product $\prod_{j=1}^n  \mathcal{L}_{ac}^{(0,\Lambda_j)}(v_j-u)$ constitutes
an irreducible representation of the Yangian. The reason is simply
that the double product together with the root solution has less
symmetries and thus can lead to additional invariant subspaces.

It is proven in \cite{molev-representations} that all finite
dimensional representations of the twisted Yangian are obtained by the
projection method from dressed solutions like \eqref{dressedK2}: to
obtain irreducible solutions one
needs to project to the subspace generated by the so-called highest
weight vector within the tensor product of the various auxiliary spaces. 
In practical terms this means that each irreducible solution is of the form
\begin{equation}
  \label{dressedK3}
   \tilde K(u)=\Pi \left(\prod_{j=1}^n  \mathcal{L}^{(0,\Lambda_j)}(v_j-u)\right)
   K_\Omega(u) \left(\prod_{j=n}^1
     (\mathcal{L}^{(0,\Lambda_j)}(v_j+u))^t\right) \Pi,
\end{equation}
where $\Pi$ is some projector acting on
$V_{\Lambda_n}\otimes\dots\otimes V_{\Lambda_1}\otimes V_\Omega$. For a graphical
interpretation see the second figure in \ref{fig:dressing}.

An equivalent interpretation of this statement is, that the dressed
MPS' can always be expressed as sums of irreducible MPS as
\begin{equation}
  \prod_{j=1}^n  t^{(\Lambda_j)}(v_j) \ket{\Psi^\Omega}=\sum_{j} \ket{\Psi_j},
\end{equation}
and every finite dimensional irreducible MPS satisfying also
\eqref{symrel} will be included in one of these sums.

Our goal in this paper is to find those solutions to the BYB that
describe our special one-site invariant MPS, listed at the end of
Section \ref{sec:iMPS}.
One possible strategy would be to investigate the general solution \eqref{dressedK3}, to
explicitly compute
the initial values at $u=0$, and to identify those solutions that
describe the MPS' with a concrete one-site block $\omega$. However, in the present paper we
use a different method: we
explicitly solve the sq.r.r. starting from a given integrable
$\omega$. By the arguments of Section \ref{sec:inte} they also satisfy
the twisted BYB and the symmetry relations, and they will
constitute an irreducible representation of the twisted Yangians.

\section{Solutions for the symmetric pair $(SU(N),SO(N))$}

\label{sec:su}

In this section we analyze the solutions of the
sq.r.r. \eqref{gyokBYB} in the case of the symmetric pair $(SU(N),SO(N))$.

In this case the $R$-matrix is given by \eqref{SUNR} and the sq.r.r. can be written in the form
\begin{equation}
  \label{rec1}
  \omega_c\psi_{ba}-\psi_{cb}\omega_a=u[\psi_{ca},\omega_b].
\end{equation}
We are looking for polynomial solutions of finite order:
\begin{equation}
  \psi_{ab}(u)=\omega_a\omega_b+\sum_{j=1}^m \psi^{(j)}_{ab}(u) u^j.
\end{equation}
The leading coefficient needs to be a scalar with $SO(N)$ symmetry, so
we have
\begin{equation}
  \label{highor}
  \psi^{(m)}_{ab}=\delta_{ab}.
\end{equation}
In the following we will assume that the $\omega$ matrices are invertible.

Eq. \eqref{rec1} can be used recursively to fix the
coefficients. We get
\begin{equation}
  \label{rec2}
   \omega_c\psi^{(j)}_{ba}-\psi^{(j)}_{cb}\omega_a=[\psi^{(j-1)}_{ca},\omega_b],
   \qquad j=1\dots m,
\end{equation}
where it is understood that $\psi^{(0)}_{ab}=\omega_a\omega_b$ and the
highest order equation corresponding to $j=m+1$ is trivially satisfied
with \eqref{highor}.

The linear operator on the l.h.s. of \eqref{rec2} has a null space,
and according to Proposition \ref{psi0} this space is one dimensional
with a basis vector given by $\psi^{(0)}_{ab}=\omega_a\omega_b$.  It
follows that if we find a concrete solution $\psi^{(j)}_{ba}$ 
the inhomogeneous equation at order $j$ then most the general solution will
be
\begin{equation}
  \psi^{(j)}_{ab}+\alpha \omega_a\omega_b,\qquad \alpha\in\complex.
\end{equation}

\begin{prop}
  The linear term is
  \begin{equation}
    \label{linear}
    \psi^{(1)}_{ab}=\omega_b\omega_a+\alpha \omega_a\omega_b
  \end{equation}
  with some $\alpha\in\complex$.
\end{prop}
\begin{proof}
  It is easy to see that the solution with $\alpha=0$ satisfies
  \eqref{rec2}, thus we get the general solution as given above.
\end{proof}

Simple solutions are obtained by restrictions on the highest degree
$m$. In the linear case we get:

\begin{prop}
If the $\omega$ matrices are invertible then the only linear solution
is (up to overall re-scaling) $\omega_a=\gamma_a$ and
\begin{equation}
  \label{gammasolution}
    \psi_{ab}(u)=\gamma_a\gamma_b+2u\delta_{ab},
  \end{equation}
  where the $\gamma$ matrices satisfy the $N$ dimensional Clifford
  algebra:
  \begin{equation}
    \{\gamma_a,\gamma_b\}=2\delta_{ab}.
  \end{equation}
\end{prop}
\begin{proof}
In this case the linear term \eqref{linear} has to coincide with
\eqref{highor}, thus
\begin{equation}
  \omega_b\omega_a+\alpha \omega_a\omega_b=\beta \delta_{ab}.
\end{equation}
The r.h.s. is symmetric with respect to $a,b$, yielding
$\alpha=1$. With a proper normalization of $\omega_a$ we can set
$\beta=2$, and thus we obtain the Clifford algebra.
\end{proof}

In this case the solution can be written alternatively as
\begin{equation}
  \psi_{ab}(u)=(2u+1)\tilde\psi_{ab}(u),\qquad
  \tilde\psi_{ab}(u)=\delta_{ab}+[\gamma_a,\gamma_b]/4(u+1/2)^{-1}.
\end{equation}
Here we recognize the generators of $SO(N)$ in the spinor
representation given by $F_{ab}=[\gamma_a,\gamma_b]/4$, thus this
solution is equal to a ``root solution'' as given by
\eqref{rootsolutions}. We
remind that the corresponding MPS is one of our main examples listed
in Section \ref{sec:iMPS}.

Now we consider the quadratic case. If $\psi_{ab}(u)$ is of second order, then for $j=2$ in \eqref{rec2}
\begin{equation}
  [\omega_a\omega_c+\alpha \omega_c\omega_a,\omega_b]=
  \beta \left[\omega_c\delta_{ba}-\delta_{cb}\omega_a\right]
\end{equation}
with some $\alpha,\beta\in\complex$, where $\beta\ne 0$ is the
coefficient of $\delta_{ab}$. The r.h.s. is
anti-symmetric with respect to $a,c$, therefore $\alpha=-1$ or
$\{\omega_a,\omega_c\}$ commutes with all $\omega_b$. The latter case
leads us back to the Clifford algebra, therefore we are free to set
$\alpha=-1$. After an exchange of indices and setting $\beta=-1$ we
obtain the exchange relations
\begin{equation}
  \label{trirel}
\Big[\omega_a,[\omega_b,\omega_c]\Big]=\delta_{ab}\omega_c-\delta_{ac}\omega_b.
\end{equation}
In this case the second order solution is
\begin{equation}
  \psi_{ab}(u)=\omega_a\omega_b-u[\omega_a,\omega_b]-u^2\delta_{ab}.
\end{equation}

We construct explicit realizations of the algebra
\eqref{trirel}. Consider first any triplet of the form
$\{\omega_j,\omega_k,\pm i[\omega_j,\omega_k]\}$, $j\ne k$. It follows from
\eqref{trirel} that they satisfy the $SU(2)$-algebra.
In those cases
when the full set $\{\omega_j\}$ is a Lie-algebra (closed with respect
to the commutator), it is easily shown that the $SU(2)$-algebra is the
only possibility, thus we obtain the solution for $N=3$ 
\begin{equation}
  \label{higherZarembo}
  \psi_{ab}(u)=S_aS_b-u[S_a,S_b]-u^2\delta_{ab},
\end{equation}
where the $S_a$ are the spin operators in a finite dimensional
irreducible representation.
This solution describes
the MPS given by $\omega_a=S_a$, which is one of our main examples
listed at the end of Section \ref{sec:iMPS}.

In those cases when the set $\{\omega_j\}$ is not a Lie-algebra (the
set is not closed with respect to the commutator), we
can consider the commutation relations of the matrices
$t_{jk}=-t_{kj}\equiv i[\omega_j,\omega_k]$:
\begin{equation}
  \begin{split}
  [t_{jk},t_{lm}]&=-\Big[[\omega_j,\omega_k],[\omega_l,\omega_m]\Big]=
   t_{jl}\delta_{km}+t_{km}\delta_{jl}- t_{jm}\delta_{kl}-t_{kl}\delta_{jm}.
 \end{split}  
\end{equation}
We see that the $t_{jk}$ satisfy the $\mathfrak{o}(N)$
Lie-algebra. Supplied with the additional conditions
\begin{equation}
  [\omega_j,t_{kl}]=i[\delta_{jk}\omega_l-\delta_{jl}\omega_k],\qquad
  [\omega_j,\omega_k]=-i t_{jk},
\end{equation}
it can be seen that the only irreducible solutions are the ``fused
spinor representations'', which are constructed as follows. Let
$\gamma_j$, $j=1\dots N$ be the $N$-dimensional Gamma matrices
satisfying the Clifford algebra. Denoting their vector space as
$V_\gamma$ we construct the symmetrized tensor product spaces
\begin{equation}
  \text{Sym} \left\{\otimes_{j=1}^n V_\gamma \right\}.
\end{equation}
Let $\Pi_n$ be the projector from $\otimes_{j=1}^n V_\gamma$ onto the
symmetrized space. Then we construct the matrices
\begin{equation}
  \label{Gamman}
  \Gamma_j^{(n)}=\Pi_n
  \left[
    \underbrace{\gamma_j\otimes 1\otimes\dots\otimes 1}_{n}+
      \underbrace{1\otimes \gamma_j\otimes\dots\otimes 1}_{n}+\dots
    \right]
  \Pi_n.
\end{equation}
Choosing $\omega_j=\Gamma^{(n)}_j/2$ we see that the commutation
relations are satisfied, thus we obtain a solution
\begin{equation}
  \label{higherGamma}
  \psi_{ab}(u)=\frac{\Gamma^{(n)}_a\Gamma^{(n)}_b}{4}-u
\frac{[\Gamma^{(n)}_a,\Gamma^{(n)}_b]}{4}-u^2\delta_{ab}.
\end{equation}
Note that this solution has the same form as \eqref{higherZarembo} with the
identification $S_a=\Gamma^{(n)}_a/2$. The reason is that the higher
spin $SU(2)$ generators are fused symmetrically from the spin-1/2
generators $\sigma_a/2$, and the Pauli matrices satisfy the Clifford algebra.

It can be checked by direct computation that these solutions satisfy
the un-normalized symmetry relation \eqref{psiRpsi}. For example the
solution \eqref{gammasolution} satisfies \eqref{psiRpsi} with $f(u)=\frac{1+2u}{1-2u}$.

Whereas \eqref{gammasolution} is a known ``root solution'' to the
twisted BYB, the explicit solutions \eqref{higherZarembo} and 
\eqref{higherGamma} are new.

\section{Solutions for $(SO(N),SO(D)\otimes SO(N-D))$}

\label{sec:so}

In this section we analyze the solutions of the
sq.r.r. \eqref{gyokBYB} in the case of the symmetric pair
$(SO(N),SO(D)\otimes SO(N-D))$, assuming irreducibility.

In this case the $R$-matrix is given by \eqref{SONR}, which is
crossing invariant by \eqref{SONcrossing}. This poses an additional
constraint which is most easily derived from the intertwining relation
\eqref{intertwK}. At the special point $u=-c/2$ we have
\begin{equation}
  R^{T}(-c/2)=R(-c/2),
\end{equation}
therefore the two sets of dressed MPS are completely identical, which
by irreducibility implies
\begin{equation}
  \label{extracond}
  K(-c/2)\sim I.
\end{equation}
Furthermore, the crossing and the intertwining relations also imply
the inversion relation
\begin{equation}
  \label{KinvSON}
  K(u)K(-u-c)\sim I.
\end{equation}

If $K(u)$ is a solution to the BYB \eqref{twisted}, then the crossed $K$-matrices defined as
\begin{equation}
  \label{kisk}
  k(u)=K(-u-c/2)
\end{equation}
satisfy the standard BYB
\begin{equation}
  \label{usual}
k_2(v) R_{21}(u+v) k_1(u)R_{12}(v-u)=
  R_{21}(v-u)k_1(u) R_{12}(u+v)k_2(v).
\end{equation}
Condition \eqref{extracond} then translates into the usual initial
condition $k(0)\sim 1$.

Instead of a systematic analysis of all possible solutions we
construct a few particular ones.

\subsection{One-site states}

Here we assume an MPS with bond dimension one, i.e. a one-site product
state given by a vector $\omega_a$.

We make the following Ansatz:
\begin{equation}
  \psi_{ab}(u)=f(u)\omega_a\omega_b+g(u) \delta_{ab},
\end{equation}
and we assume
\begin{equation}
    \label{oman}
    \omega_a\omega_a=1.
\end{equation}
The sq.r.r. reads:
\begin{equation}
  \begin{split}
  (u+c)\omega_c\psi_{ba}(u)+u(u+c)\omega_b \psi_{ca}(u)-u
  \delta_{cb} \omega_d \psi_{da}(u)=\\
 (u+c) \psi_{cb}(u)\omega_a+u(u+c) \psi_{ca}(u)\omega_b-u
 \delta_{ba} \psi_{cd}(u)\omega_d.
  \end{split}
\end{equation}

This gives the solution
\begin{equation}
  f(u)=2u+c\qquad g(u)=-u.
\end{equation}
Allowing an arbitrary normalization for $\omega$ we get
\begin{equation}
  \label{SON1sitestate}
  \psi_{ab}(u)=(2u+c) \omega_a\omega_b-u (\omega_d\omega_d) \delta_{ab}.
\end{equation}
This shows that in the limit of $\omega_d\omega_d\to 0$ the second term
decouples and thus the first one can be chosen as a constant.

After constructing the  $K$ matrix as \eqref{Kpsi} it is easy to see
that the inversion relation \eqref{KinvSON} is satisfied with the following proportionality
factor:
\begin{equation}
  K(u)K(-u-c)=(-u)(u+c).
\end{equation}

\subsection{Two-site states}

For sake of completeness we give here the general scalar valued solution of the twisted BYB
with the given symmetry properties. This solution describes a two-site
state, therefore the sq.r.r. can not be applied here. 

The solutions of the BYB relation are classified in
\cite{SON-R-matrices-classification-MORICONI} and summarized for
example in \cite{Gombor-SON-nonstandard}. In our conventions the
relevant solution can be written as
\begin{equation}
  K(u)=
  \begin{pmatrix}
   (N-D-1+2u) I_{D}& 0\\
    0 & (-D+1-2u) I_{N-D} 
  \end{pmatrix},
\end{equation}
where $I_D$ and $I_{N-D}$ stand for identity matrices of dimension
$D$, and $N-D$, respectively. It is easy to see that the inversion
relation \eqref{KinvSON} is satisfied.

In the special case of $D=1$ we get (a scalar multiple of) the one-site
state solution given in \eqref{SON1sitestate} with the vector $\omega=(1,0,\dots,0)$.

\subsection{Matrix Product States}

To simplify notations we split the full vector space $\complex^N$ into a direct
sum of $\complex^D$ and $\complex^{N-D}$.
We will use indices $a,b,\ldots=1\dots D$ describing the coordinates
in the first component, and $I,J,\ldots =1\ldots (N-D)$ for the second.

We investigate the MPS given by
\begin{equation}
  \omega_a=\gamma_a \qquad \omega_I=0,
\end{equation}
where the $\gamma_a$ are the $D$-dimensional Gamma matrices.

We make the following Ansatz for the solution of the sq.r.r.:
\begin{equation}
  \begin{split}
    \psi_{ab}&=f(u)\gamma_a\gamma_b+g(u)\delta_{ab}\\
    \psi_{aI}&=\psi_{Ia}=0\\
   \psi_{IJ}&=h(u)\delta_{IJ}.
    \end{split}
  \end{equation}

  We now investigate the different components of the sq.r.r.
  If all indices are in the first subgroup then we get
\begin{equation}
  \begin{split}
  (u+c)\omega_c\psi_{ba}(u)+u(u+c)\omega_b \psi_{ca}(u)-u
  \delta_{cb} \omega_d \psi_{da}(u)=\\
 (u+c) \psi_{cb}(u)\omega_a+u(u+c) \psi_{ca}(u)\omega_b-u
 \delta_{ba} \psi_{cd}(u)\omega_d.
  \end{split}
\end{equation}

Substituting our Ansatz and observing the cancellation of a few terms
we get
\begin{equation}
  \begin{split}
  &  (u+c)
\delta_{ab}\gamma_cg(u)
    +u(u+c)\gamma_b\gamma_c\gamma_a f(u)-u
      \delta_{cb} \gamma_d  (f(u)\gamma_d\gamma_a+g(u)\delta_{ad})=
\\ &      (u+c) \delta_{bc}      \gamma_a  g(u)
      +u(u+c)\gamma_c\gamma_a     \gamma_b   f(u)
      -u \delta_{ba}  (f(u)\gamma_c\gamma_d+g(u)\delta_{cd})     \gamma_d.
  \end{split}
\end{equation}
We can use $\gamma_j\gamma_j=D$ to obtain
\begin{equation}
  \begin{split}
&\left(    \delta_{ab}\gamma_c-   \delta_{bc}      \gamma_a\right)\left( (2u+c)g(u)+Duf(u)\right)-\\
&
  +u(u+C)(\gamma_b\gamma_c\gamma_a-\gamma_c\gamma_a     \gamma_b)f(u)=0.
\end{split}
\end{equation}
Making use of the Clifford algebra relations we end up with 
\begin{equation}
\left(    \delta_{ab}\gamma_c-   \delta_{bc}
  \gamma_a\right)\left( (2u+c)g(u)+(-2u^2-2cu+Du)f(u)\right)
=0.
\end{equation}
A solution is 
\begin{equation}
  f(u)=2u+c\qquad g(u)=2u^2+(-D+2c)u.
\end{equation}

For the remainder we only need to check cases when two out of three
indices belong to $IJ$, because the other possibilities are automatically zero.

When the indices are specified as $cIJ$ we get
\begin{equation}
(u+C)\gamma_c\delta_{IJ}h(u)=-u\delta_{IJ}
(f(u)\gamma_c\gamma_d+g(u)\delta_{cd}) \gamma_d  =
-u\delta_{IJ}(Df(u)+g(u))\gamma_c
\end{equation}
giving
\begin{equation}
  h(u)=-u(D+2u).
\end{equation}
It is easy to see that all other cases are satisfied are well. The
solution is thus
\begin{equation}
  \begin{split}
    \psi_{ab}(u)&=(2u+c)  \gamma_a\gamma_b+(2u^2+(-D+2c)u)\delta_{ab}\\
    \psi_{aI}(u)&=\psi_{Ia}(u)=0\\
   \psi_{IJ}(u)&=-u(D+2u)\delta_{IJ}.
    \end{split}
  \end{equation}
It is easy to see that the condition \eqref{extracond} is indeed satisfied.
  
Setting $D=N$ and substituting $c=N/2-1$ we obtain the completely
$SO(N)$-invariant solution
\begin{equation}
     \psi_{ab}(u)=(2u+N/2-1)  \gamma_a\gamma_b+(2u^2-2u)\delta_{ab}.
\end{equation}
  
In the special case of $D=3$ (and arbitrary $N$) the Clifford generators are given by the
Pauli matrices. For $N=6$
we get a particular solution
\begin{equation}
  \begin{split}
    \psi_{ab}(u)&=(2u+2)  \sigma_a\sigma_b+(2u^2+u)\delta_{ab}\\
    \psi_{aI}(u)&=\psi_{Ia}(u)=0\\
   \psi_{IJ}(u)&=-u(3+2u)\delta_{IJ}.
    \end{split}
  \end{equation}
This solution is relevant to one-point functions in AdS/CFT
\cite{adscft-kristjansen-2017,kristjansen-proofs}. After the crossing transformation \eqref{kisk} we get a
particular solution $k(u)$ to the untwisted BYB, which was obtained earlier by
DeWolfe and Mann in \cite{dewolfe-mann}.

Direct computation shows that these solutions satisfy
the un-normalized symmetry relation \eqref{psiRpsi}. For example the simplest
solution \eqref{SON1sitestate} satisfies \eqref{psiRpsi} with $f(u)=\frac{c+2u}{c-2u}$.

\section{Conclusions}

\label{sec:conclusions}

In this work we treated integrable MPS and established a new linear
intertwining relation \eqref{intertwK}-\eqref{gyokBYB} called the ``square root relation'' which
guarantees the integrability property. We showed that under
certain conditions the solutions to the sq.r.r. also solve the
twisted BYB relation \eqref{twisted}. The sq.r.r. is only linear as
opposed to the quadratic BYB, thus it is much easier to solve.
From the general point of view of boundary integrability, we can
expect that any solution of the twisted BYB which factorizes at the
special point $u=0$ (see the initial condition \eqref{init}) also
solves the sq.r.r. with the corresponding one-site object $\omega$.

The question of under which circumstances the three different relations (the integrability
condition, and the sq.r.r. and twisted BYB) are completely equivalent
remains open. The connection between these three properties is
depicted in Fig. \ref{fig:rel}.
We have shown that the three conditions are equivalent if certain
dressed MPS are irreducible. However, there are solutions
which satisfy all three conditions without the complete reducibility (for example
the reference state in the $SU(N)$-invariant model), so perhaps this
condition can be weakened. It would be desirable to clarify
this issue. Also, it would be important to develop analytic proofs for
the irreducibility condition, which we confirmed only numerically in a
few concrete cases. 
The reducibility properties of the dressed MPS are
related to the reducibility of the fused representations of the twisted
Yangian, for which there are no general results available. These
particular problems deserve further work. 

\begin{figure}[]
  \centering
  \begin{tikzpicture}[scale=0.7]
    \node (a) at (0,0) {\footnotesize Integrability conditions \eqref{int1}-\eqref{int2}};

\node (b) at (7.5,2) {\footnotesize Square root relation \eqref{gyokBYB}};

\node at (7.5,-2) {\footnotesize BYB relation \eqref{BYB}};

\draw [thick,->] plot [smooth] coordinates { (4.5,2)  (3,1.8) (0,0.5)};
\draw [thick,->] plot [smooth] coordinates { (5,-2)  (3,-1.8) (0,-0.5)};

\draw [thick,dashed,->] plot [smooth] coordinates {(3.8,0.2) (5,1)  (6,1.7)};
\draw [thick,dashed,->] plot [smooth] coordinates {(3.8,-0.2) (5,-1) (6,-1.7)};

\draw [thick,dashed,->] plot [smooth] coordinates {(7.5,-1.7) (8,0)  (7.5,1.7)};
\draw [thick,dashed,->] plot [smooth] coordinates {(7,1.7) (6.5,0) (7,-1.7)};
\end{tikzpicture}
  \caption{Relation between the different notions of integrability for
    a one-site invariant MPS. The thick lines stand for direct implication,
    whereas the dashed lines show implication given the irreducibility
  condition of the dressed MPS, as explained in the main text.}
  \label{fig:rel}
\end{figure}
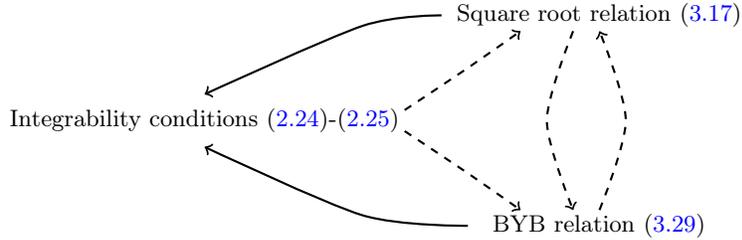

In Sections \ref{sec:su}-\ref{sec:so} we have provided solutions of the sq.r.r. for the integrable MPS that
were listed at the end of Section \ref{sec:iMPS}.
Moreover, in \ref{sec:su} we classified all solutions to the
sq.r.r. having the symmetry $(SU(N),SO(N))$ which are at most
quadratic polynomials in the rapidity. The linear solutions
necessarily lead the Clifford algebra, whereas the quadratic solutions
correspond to symmetrically fused Clifford generators.
On the other hand, this list does not
exhaust all known integrable MPS: the fused states for the pair
$(SO(N),SO(D)\otimes SO(N-D))$ (see
\cite{adscft-kristjansen-2017,ads-neel-cite-kristjansen-so5,kristjansen-proofs}
for details) were not treated here.

Having found our solutions to the sq.r.r., and having established that
they also solve the twisted BYB there are the following open directions.

First, one should establish the explicit fusion hierarchy of
these solutions, which could lead to an understanding of the physical
overlaps using the combination of the Quench Action and QTM
methods. These computations would be analogous to the ones of
\cite{sajat-minden-overlaps,sajat-su3-2}. 

Second, one should understand better the explicit relations with the
representation theory of the twisted Yangian. In particular, it would
be useful to understand the exact decomposition of the form
\eqref{dressedK3}, together with the projector $\Pi$. Preliminary
computations show that our solution \eqref{higherZarembo} can always be
obtained by a single fusion of LKL type and a non-trivial
projection. Determining the decomposition could also help in
understanding the overlap formulas obtained in \cite{ADSMPS2,adscft-kristjansen-2017,kristjansen-proofs}.

Finally, it would be interesting to obtain the spectrum and the
eigenstates of the double row transfer matrices with the new boundary
$K$-matrices. As we remarked, these transfer
matrices can lead to new integrable models with local Hamiltonians and additional
boundary degrees of freedom. 

We plan to return to these questions in further research.

\vspace{1cm}
{\bf Acknowledgments} 

\bigskip

We would like to thank Tam\'as Gombor, Ben Grossmann, Charlotte Kristjansen, Marius de Leeuw,
Georgios Linardopoulos, 
Vidas Regelskis, G\'abor Tak\'acs, Kasper
Vardinghus and Matthias Wilhelm for useful discussions. BP was supported
by the BME-Nanotechnology FIKP grant of EMMI (BME 
FIKP-NAT), by the National Research Development and Innovation Office
(NKFIH) (K-2016 grant no. 119204 and KH-17 grant no. 125567), and by
the Premium Postdoctoral Program of the Hungarian Academy of
Sciences. EV acknowledges support by the EPSRC under grant
EP/N01930X. 

\appendix

\section{The different forms of the square root relation}

\label{app:basszus}

Here we provide a detailed derivation showing that
\eqref{gyokBYB00} follows from \eqref{intertwK}. 

We start with the intertwining relation
 \begin{equation}
    A_jK(u)=K(u)B_j,\quad j=1,\dots,N,
  \end{equation}
  which is a relation in $End(V_0\otimes V_A)$ and it concerns the
  dressed MPS defined as
\begin{equation}
  \begin{split}
    A_j&=\sum_k \mathcal{L}_{jk}(u)\otimes \omega_k \\
    B_j&=\sum_k \mathcal{L}^T_{jk}(u)\otimes \omega_k. \\
  \end{split}
\end{equation}
Here the Lax operators $\mathcal{L}_{jk}(u)$ are defined from the
expansion of the $R$-matrix as
\begin{equation}
  R_{10}(u)=\sum_{ab} E_{ab} \otimes \mathcal{L}_{ab}(u),
\end{equation}
where $E_{ab}$ are the basis matrices acting on a physical space
$V_1\approx \complex^N$ and the matrices
$\mathcal{L}_{ab}(u)$ act on the auxiliary space $V_0$. This expansion
also implies that the matrix elements of the $R$-matrix are
\begin{equation}
  \label{nami}
  R_{bc}^{ad}(u)= (\mathcal{L}_{ab}(u))_{c}^d=\text{Tr}_0\big( E_{dc}\mathcal{L}_{ab}\big).
\end{equation}
In our cases the $R$-matrix is symmetric with respect to an exchange
of the two vector spaces, and also with respect to the exchange of the
in- and out-states, therefore
\begin{equation}
   R_{bc}^{ad}(u)= R_{cb}^{da}(u)= R^{bc}_{ad}(u).
\end{equation}

We also expand the $K$-matrix as
\begin{equation}
K(u)=\sum_{a,b}E_{ab}\otimes \psi_{ab}(u),
\end{equation}
where $E_{ab}$ are elementary matrices acting on $V_0$ and
$\psi_{ab}(u)$ are matrices acting on $V_A$.

Then for every $j=1\dots N$ we get
\begin{equation}
\sum_{a,b,k}
 \Big(\mathcal{L}_{jk}(u)E_{ab}\Big) \otimes \Big(\omega_k  \psi_{ab}(u)\Big)=
\sum_{a,b,k}\Big(E_{ab}\mathcal{L}^T_{jk}(u)\Big)\otimes \Big(\psi_{ab}(u)\omega_k\Big).
\end{equation}
We multiply this equation with $E_{cd}\otimes 1$, where again
$E_{cd}$ acts on $V_0$ and $c,d=1\dots N$ are unspecified indices. Fruthermore, we take the partial trace over
$V_0$ to arrive at
\begin{equation}
\sum_{a,b,k}
 \text{Tr}\big(E_{cd}\mathcal{L}_{jk}(u)E_{ab}\big) \Big(\omega_k  \psi_{ab}(u)\Big)=
\sum_{a,b,k}\text{Tr}\big(E_{cd}E_{ab}\mathcal{L}^T_{jk}(u)\big) \Big(\psi_{ab}(u)\omega_k\Big).
\end{equation}
This can be written as
\begin{equation}
\sum_{a,k}
 \text{Tr}\big(\mathcal{L}_{jk}(u)E_{ad}\big) \Big(\omega_k  \psi_{ac}(u)\Big)=
\sum_{b,k}\text{Tr}\big(\mathcal{L}_{jk}(u)E_{bc}\big) \Big(\psi_{db}(u)\omega_k\Big),
\end{equation}
where we also used $E_{ab}^T=E_{ba}$ for every $a,b=1\dots N$.

Using \eqref{nami} and the symmetries of the $R$-matrix this gives
\begin{equation}
  \sum_{a,k}
R_{jd}^{ka}(u)
  \Big(\omega_k  \psi_{ac}(u)\Big)=
  \sum_{b,k}
R^{kb}_{jc}(u)
  \Big(\psi_{db}(u)\omega_k\Big).
\end{equation}
Introducing also the matrix $\check R(u)=PR(u)$ this is written as
\begin{equation}
  \sum_{a,k}
\check R_{dj}^{ka}(u)
  \Big(\omega_k  \psi_{ac}(u)\Big)=
  \sum_{b,k}
\check R^{bk}_{jc}(u)
  \Big(\psi_{db}(u)\omega_k\Big).
\end{equation}
After an exchange of indices this is indeed equivalent to
\eqref{gyokBYB00}.

\section{Numerical tests of the irreducibility condition}

\label{sec:C1}

Here we describe our simple numerical procedure to test the irreducibility
condition for the dressed MPS defined in \eqref{ABdef}.

Given a set of $d\times d$ matrices $\{A_j\}_{j=1\dots N}$ the goal is to test
whether the linear span of all matrix products of length $L$ becomes
$End(\complex^d)$ for $L\ge L^*$. This can be tested recursively in
$L$. At $L=1$ we investigate the matrices themselves and construct an
orthonormal basis using the scalar product
\begin{equation}
  (A,B)\equiv \text{Tr} A^\dagger B.
\end{equation}
Let us denote this basis by $\{B^{(1)}_j\}_{j=1\dots n_1}$, where $1\le
n_1\le N$. Now we construct the basis at length $L=2$ by computing all
products
\begin{equation}
  A_j B^{(1)}_k,\quad j=1\dots N,\quad k=1\dots n_1
\end{equation}
and performing a new orthogonalization procedure. This gives a basis
$\{B^{(2)}_j\}_{j=1\dots n_2}$, where typically $n_2>n_1$. We continue this
procedure such that at each step we define
\begin{equation}
  \{B^{(L+1)}_j\}_{j=1\dots n_{L+1}}\equiv \text{ basis of }span\{A_j
  B^{(L)}_k\}_{j=1\dots N,k=1\dots n_L}.
\end{equation}
Generally we observe that $n_L$ is always increasing, and the MPS is
irreducible if $n_L$ it reaches
the maximal value $d^2$ at some $L=L^*$. 

We performed this test for a few examples of the dressed matrices \eqref{ABdef} constructed from the
integrable MPS listed at the end of Section \ref{sec:iMPS}; we took at
least one example from each family.
We observed that the dressed MPS are irreducible for generic
values of the rapidity $u$. On the other hand, there are some special 
rapidity points (for example $u=0$), when $n_L$ saturates below $d^2$
and thus the MPS has invariant
subspaces. This does not affect the applicability of Theorem \ref{simmm},
because the existence of the unique integrable $K$-matrices is
established for almost all $u$, and the intertwiner
relation \eqref{intertwK} holds even at the special points by
continuity.

We observed that the set of special rapidity points for which
$n_L<d^2$ always consists of integers. The point $u=0$ is always
included in this set, which can be tied to the special value of the
$R$-matrix $R(0)=P$ used in the dressing \eqref{ABdef}. However,
depending on $\{\omega_j\}$ there can be other reducible points. For
example in the case of $(SU(3),SO(3))$ and $\omega_j=S_j$ with the
$S_j$ being the spin-1 generators we find that the reducible points
are $u=\{-2,0,1\}$.

We also investigated the doubly dressed matrices \eqref{double}. We
observed, that for generic rapidity parameters $u,v$ the MPS are
irreducible in all our examples. This implies that the 
solutions of the sq.r.r. also solve the twisted BYB relation. Once
again we observed special points where the MPS have invariant
subspaces, but this does not alter the conclusions.

\section{The square root relation in the XXZ chain}

\label{sec:XXZ}

Here we prove that the known solutions of the BYB in the XXZ chain
also satisfy the sq.r.r., given that we choose those solutions which
describe one-site states after the crossing transformation. This way
we can also prove that all one-site states are integrable in the XXZ
chain, even in odd volumes \cite{sajat-minden-overlaps}.

In the case of the XXZ chain the crossing transformation for the
boundary two-site block is \cite{sajat-integrable-quenches}
\begin{equation}
  \psi_{ab}(u)=(K(-\eta/2-u)\sigma_y)_{ab}.
\end{equation}
The $K$-matrices are given by a known 3-parameter family \cite{general-K-XXZ}. Choosing a
particular parametrization with constants $\alpha,\beta,\theta$ they lead to
the following form of the two-site block $\psi_{ab}(u)$:
\begin{equation}
  \label{psyparam}
\begin{split}
 \psi_{11}(u)=& -e^\theta \sinh(\eta+2u) \\
  \psi_{12}(u)=& 2(-\sinh(\alpha)\cosh(\beta)\cosh(\eta/2+u)+ \cosh(\alpha)\sinh(\beta)\sinh(\eta/2+u))\\
 \psi_{21}(u)=& 2(\sinh(\alpha)\cosh(\beta)\cosh(\eta/2+u)+ \cosh(\alpha)\sinh(\beta)\sinh(\eta/2+u))\\
 \psi_{22}(u)=&e^{-\theta}\sinh(\eta+2u).
\end{split}
\end{equation}
One-site states are obtained by choosing $\alpha=0$ and
$\beta=i\pi/2+\eta/2$. Dividing by $\sinh(\eta)$ we obtain the
two-site block
\begin{equation}
  \label{psyparam1a}
\begin{split}
 \psi_{11}(u)=& -e^\theta \frac{\sinh(\eta+2u)}{\sinh(\eta)} \\
  \psi_{12}(u)=&  \psi_{21}(u)=i\frac{\sinh(\eta/2+u)}{\sinh(\eta/2)}\\
 \psi_{22}(u)=&e^{-\theta}\frac{\sinh(\eta+2u)}{\sinh(\eta)}\\
\end{split}
\end{equation}
satisfying the initial condition
 \begin{equation}
  \psi(0)=\omega\cdot \omega,\quad \text{with}\quad\omega=
  \begin{pmatrix}
ie^{\theta/2}  \\     e^{-\theta/2}
  \end{pmatrix}.
\end{equation}
It can be checked by direct computation that this $\psi(u)$ solves the
sq.r.r. with the $\omega$ given above.

\addcontentsline{toc}{section}{References}
\providecommand{\href}[2]{#2}\begingroup\raggedright\endgroup

\end{document}